\documentclass[a4paper,11pt]{article}
\usepackage[utf8]{inputenc}
\usepackage{geometry}

\usepackage{setspace}

\usepackage{amsmath}
\usepackage{amsfonts,amssymb,amsthm}
\usepackage{xcolor}
\definecolor{linkColor}{RGB}{0, 128, 128}
\definecolor{citeColor}{RGB}{0, 112, 64}
\definecolor{urlColor}{RGB}{120, 0, 120}
\usepackage{hyperref}
\hypersetup{
    colorlinks=true,
    linkcolor=linkColor,
    citecolor=citeColor,      
    urlcolor=urlColor,
}
\usepackage{url}

\definecolor{cmdColor}{RGB}{0, 0, 0}

\theoremstyle{plain}
\newtheorem{thm}{Theorem}
\newtheorem{lem}[thm]{Lemma}

\newtheorem{cor}[thm]{Corollary}
\newtheorem{clm}[thm]{Claim}
\newtheorem{fact}[thm]{Fact}

\newtheorem{defn}[thm]{Definition}

\usepackage{tikz}
\usetikzlibrary{quantikz} 

\usepackage{youngtab}

\newcommand{\cA}{\mathcal{A}}
\newcommand{\cB}{\mathcal{B}}

\newcommand{\cF}{\mathcal{F}}
\newcommand{\cH}{\mathcal{H}}
\newcommand{\cR}{\mathcal{R}}
\newcommand{\cS}{\mathcal{S}}
\newcommand{\cV}{\mathcal{V}}

\renewcommand{\>}{\rangle}
\newcommand{\<}{\langle}
\newcommand{\Sym}{\mathbb{S}}
\renewcommand{\S}{\mathbb{S}}

\newcommand{\spA}[1]{\mathcal{A}_{#1}}
\newcommand{\spB}[1]{\mathcal{B}_{#1}}

\DeclareMathOperator{\tr}{\mathrm{Tr}}
\DeclareMathOperator{\Tr}{\mathrm{Tr}}
\DeclareMathOperator{\spn}{\mathrm{span}}
\DeclareMathOperator{\sgn}{\mathrm{sgn}}

\newcommand{\Piooo}[3]{\overset{#1}{\Pi}{}^{#2}_{#3}}
\newcommand{\Xiooo}[3]{\overset{#1}{\Xi}{}^{#2}_{#3}}
\newcommand{\Hooo}[3]{\overset{#1}{\mathcal{H}}{}^{#2}_{#3}}

\newcommand{\MOp}{\Phi}

\newcommand{\Func}{{\textcolor{cmdColor}{F}}} 

\newcommand{\Dom}{{\textcolor{cmdColor}{D}}} 
\newcommand{\Ran}{{\textcolor{cmdColor}{R}}} 

\newcommand{\vn}{{\textcolor{cmdColor}{v}}} 

\newcommand{\regA}{\mathsf{A}}

\newcommand{\regF}{\mathsf{F}}
\newcommand{\regI}{\mathsf{I}}
\newcommand{\regO}{\mathsf{O}}
\newcommand{\regR}{\mathsf{R}}

\newcommand{\regW}{\mathsf{W}}
\newcommand{\regf}{\mathsf{f}}

\newcommand{\Fzero}{{\widehat{0}}} 
\newcommand{\FE}{{\widehat{E}}} 
\newcommand{\Erest}{{E_{1\hspace{-.5pt}\hat1}}} 
\newcommand{\vecr}{{r_0}} 

\newcommand{\ketPerp}{|\!\!\perp\>} 

\newcommand{\OO}{\mathrm{O}}

\newcommand{\labG}{{\mathrm{high}}}
\newcommand{\labB}{{\mathrm{low}}}
\newcommand{\hi}{{\mathrm{high}}}
\newcommand{\lo}{{\mathrm{low}}}

\newcommand{\yover}{{\textcolor{cmdColor}{\succ}}}
\newcommand{\yunder}{{\textcolor{cmdColor}{\prec}}}

\newcommand{\lightPart}[1]{{\textcolor{cmdColor}{Y_{#1}}}}

\newcommand{\zero}{{\textcolor{cmdColor}{0}}} 
\newcommand{\tabl}{{\mathfrak{T}}} 

\newcommand{\fxd}{{\textcolor{cmdColor}{\mathfrak{F}}}} 

\newcommand{\dm}[1]{{\textcolor{cmdColor}{d_{#1}}}} 

\newcommand{\Spe}[1]{{\textcolor{cmdColor}{\cS}^{#1}}} 

\newcommand\parA{\lambda}
\newcommand\parB{\mu}
\newcommand\parC{\nu}
\newcommand\parD{\zeta}
\newcommand\parE{\eta}
\newcommand\parF{\theta}

\newcommand\parEb{{\bar\parE}}

\newcommand{\Res}{\mathop{\mathrm{Res}}}

\title{Tight Bounds for Inverting Permutations\\ via Compressed Oracle Arguments}
\author{
Ansis Rosmanis\thanks{E-mail: \texttt{rosmanis@math.nagoya-u.ac.jp}}\\ [.5ex]
\normalsize  Graduate School of Mathematics \\ 
\normalsize Nagoya University
}
\date{January 20, 2022}

\begin{document}

\maketitle

\begin{abstract}
In his seminal work on recording quantum queries [Crypto 2019], Zhandry studied interactions between quantum query algorithms and the quantum oracle corresponding to random functions.
Zhandry presented a framework for interpreting various states in the quantum space of the oracle as databases of the knowledge acquired by the algorithm and used that interpretation to provide security proofs in post-quantum cryptography.

In this paper, we introduce a similar interpretation for the case when the oracle corresponds to random permutations instead of random functions.
 Because both random functions and random permutations are highly significant in security proofs, we hope that the present framework will find applications in quantum cryptography.
Additionally, we show how this framework can be used to prove that the success probability for a $\kappa$-query quantum algorithm that attempts to invert a random $N$-element permutation is at most $\OO(\kappa^2/N)$. 
\end{abstract}

\section{Introduction}

The computational model of quantum query algorithms plays an important role in security analysis of various cryptographic protocols.
To prove security of such protocols, one often has to show that a given quantum query problem is difficult to solve.
Quantitatively, in cryptography, one is typically interested in providing an upper bound on the success probability of any algorithm solving the problem as a function of the number of queries made by the algorithm.
This is closely related to the quantum query complexity of the problem, where one is interested in the required number of queries to solve the problem with probability close to $1$.

In this paper, our main problem of interest is the problem of inverting a random $N$-element permutation.
For this problem, suppose there is a known set of $N$ elements with a special element $\zero$ among them%
\footnote{For example, if $N$ is a power of two, one may consider the set of all $\log_2\!N$-bit strings with $\zero$ being the all-zeros string.}
and there is a permutation $f$ of these elements that is initially completely unknown to the algorithm,
yet completely known to an oracle. 
The algorithm can query the oracle and ask, for any element $x$, what is $f(x)$.
The task of the algorithm is to evaluate the inverse of $f$ on the special element $\zero$, namely, to return the unique preimage of $\zero$ under $f$.

The complexity of inverting permutations is well understood for about two decades. On one hand, $\kappa$ iterations of Grover's algorithm~\cite{grover:search}, each iteration requiring two calls to the oracle, find $f^{-1}(\zero)$ with success probability $\Omega(\kappa^2/N)$.
On the other hand, in the paper that pioneered the quantum adversary bound~\cite{ambainis:adversary}, Ambainis showed that inverting an $N$-element permutation with success probability $\Omega(1)$ requires $\Omega(\sqrt{N})$ calls to the oracle.
Consequently, because of the amplitude amplification~\cite{AmplHB,Grover1998}, if an algorithm makes $\kappa$ oracle calls, then its success probability is at most $\OO(\kappa^2/N)$.

In this paper, we reprove the tight $\OO(\kappa^2/N)$ upper bound on the success probability using an approach that closely relates to Zhandry's compressed oracle technique~\cite{zhandry:record}.
Zhandry analyzed the complexity of various query problems with a random oracle by, first, observing that one can assume that the oracle, instead of being uniformly random, is initialized in the uniform superposition.
He then showed that the algorithm cannot succeed unless its execution establishes a certain type of entanglement between quantum registers registers of the algorithm and those of the oracle.
In other words, by interacting with the algorithm, the reduced quantum state of the oracle has to become mixed and to have a large overlap on a certain subspace.

For proving lower bounds on quantum query complexity, this approach is not new.
The observation that a successful algorithm must establish entanglement of certain form and amount with the oracle registers is, explicitly or implicitly, the basis of the adversary method, its original form~\cite{ambainis:adversary}, its various generalizations (e.g., \cite{spalek:mult,amrr:index,lee:stateConversion,belovs2015variations}),  and its adaptations for  specific problems (e.g., see \cite{ambainis:newLowerBoundMethod}).
What is new to Zhandry's work and many of its follow-up works (see~\cite{4RoundLuby19,tcc-2019-29987,Czajkowski2019}) is that they successfully use this approach to provide security proofs in post-quantum cryptography.

Zhandry's original work addressed random oracles that for each $x$ assign the value $f(x)$ uniformly and independently at random.
Hamoudi and Magniez generalized this to the case where each $f(x)$ is still assigned independently, but not necessarily uniformly, at random~\cite{Hamoudi2020}.
However, if the oracle implements a random permutation, there cannot be such an independence, as one must have $f(x)\ne f(x')$ whenever $x\ne x'$. We overcome this obstacle by first analyzing the unstructured search problem along the same lines as by Zhandry, then introducing a fine-grain version of that analysis, and finally adapting it for the problem of inverting permutations.

Regarding the problem of inverting permutations, note that the memory of a quantum oracle that has full knowledge of an $N$-element permutation naturally corresponds to $N!$-dimensional Euclidean space. 
We decompose this space as a direct sum of about $2N$ mutually orthogonal subspaces, where each subspace essentially describes how many values $f(x)$ the algorithm has learnt and whether $\zero$ is one of them.
Our proof proceeds by showing that, in order for the algorithm to succeed, the reduced state of the oracle must have a large overlap with spaces where $\zero$ is one of the learnt values, yet a single query cannot increase this overlap by much.
While our bound $\OO(\kappa^2/N)$ on the success probability is not new, we hope that it may lead to quantum security proofs for various cryptographic constructions concerning random permutations. 

One potential application is providing a proof for the post-quantum security of the sponge construction~\cite{sponge07}.
In particular, a recent draft \cite{unruh:perm} by Unruh introduced some ideas that might lead to a proof of the collision-resistance of the sponge construction and that requires analysis of the query complexity of inverting permutations along the lines of Zhandry's compressed oracle technique.
While results required by Unruh's proposal are stronger than what we prove here (for example, we do not address ``two-sided zero search'' considered by Unruh), we hope that the present analysis will contribute towards establishing such a security proof.
We also note that a recent work~\cite{cryptoeprint:2021:192} by Czajkowski claims to prove quantum indifferentiability of the sponge construction.

It is known that the random oracle model implies the ideal cipher model in the classical setting~\cite{cryptoeprint:2008:246}. 
Another potential application would be to prove that the same holds in the quantum setting.

\paragraph{Organization of the paper.}

The paper is organized as follows.
In Section~\ref{sec:queryModel} we describe the model of quantum query algorithms and describe how to conveniently model their interaction with a random oracle. In Section~\ref{sec:dbSubspaces}, we describe how certain states of the quantum system of the oracle can be interpreted as the algorithm having learnt certain information about the permutations. We provide a formal backing for this interpretation by using it to state the main technical theorem of the paper, Theorem~\ref{thm:mainThm}, whose direct corollary are the tight bounds on the success probability for the problems of Inverting Permutations and Unstructured Search. We provide the framework for the proof of Theorem~\ref{thm:mainThm} in Section~\ref{sec:MainTmFrame}.

In Section~\ref{sec:ZhandryApproach}, we further elaborate on the connection between the present work and techniques used in \cite{zhandry:record}. We also conclude the proof of Theorem~\ref{thm:mainThm} for Unstructured Search.
The proof of Theorem~\ref{thm:mainThm} for Inverting Permutations uses the representation theory of the symmetric group.
Section~\ref{sec:RepPrelim} presents preliminaries of the representation theory used in subsequent sections.
In Section~\ref{sec:ABviaRepTh}, we show how to describe various subspaces of the database system via representation theory.
Finally, in Section~\ref{sec:MainTmProof}, we conclude the proof of Theorem~\ref{thm:mainThm}.

\section{Quantum Query Model}
\label{sec:queryModel}

Let $f\colon\Dom\rightarrow\Ran$ be a function of domain $\Dom$ and range $\Ran$, and suppose we have designated a special element $\zero\in\Ran$. Let $M:=|\Dom|$ be the size of the domain and $N:=|\Ran|$ be the size of the range. We assume that $N$ is a power of $2$ and $\Ran$ is the set of all $\log_2\!N$-bit strings.

We will consider $f$ being either a random function or a random permutation. In the permutation case,  we have $M=N$. While for permutations we also have $\Dom=\Ran$, it will often be convenient to differentiate whether we are considering the domain or the range of the permutation.
Hence, in the permutation case, we only need to assume that $\Dom$ and $\Ran$ are of the same size $N$ and that $f$ is a bijection. 

We generally use $x,x',x_i$ et cetera to refer to elements of $\Dom$ and $y,y',y_i$ et cetera to refer to elements of $\Ran$.
When we are summing over variables denoted by $x$, these variables take all values in $\Dom$. For example, we use notation $\sum_x$ instead of $\sum_{x\in\Dom}$. When such a variable is not fixed in an expression (for example, as $x$ in the rightmost norm in (\ref{eq:NormSimpl})), any value in $\Dom$ can be considered.
We do the same for variables $y$ in $\Ran$.

\subsection{Generic Quantum Query Algorithm}

Suppose a quantum algorithm is given an oracle access to $f$.
 More precisely, the algorithm has an access to a unitary transformation $O_f$ that on a given input $x$ evaluates the function $f$, namely,
\[
O_f:=\sum_{x,y'}|x\>\<x|_\regI\otimes |y'\oplus f(x)\>\<y'|_\regO,
\]
where $\oplus$ denotes the bitwise addition modulo $2$; see Figure~\ref{fig:oracle}.
The subscripts $\regI$ and $\regO$ stand for the \emph{input register} and the \emph{output register}, respectively.
We may omit these and similar subscripts when the relevant registers are clear from the context.

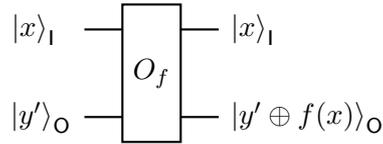
\begin{figure}[ht!]
	\centering
\begin{tikzcd}
\lstick{$\ket{x}_{\regI}$\hspace{6.2pt}} & \gate[wires=2]{O_f} &  \qw\rstick{$\ket{x}_{\regI}$} \\
\lstick{$\ket{y'}_{\regO}$\hspace{0.4pt}} & & \qw\rstick{$\ket{y'\oplus f(x)}_{\regO}$} \\
\end{tikzcd}
\caption{\small The circuit diagram of the oracle.}
\label{fig:oracle}
\end{figure}

A query algorithm, in addition to the input and output registers, has two additional registers, the \emph{workspace register} and the \emph{result register}, denoted $\regW$ and $\regR$, respectively. 
When we analyze the action of the oracle, we group $\regW$ and $\regR$ together in a larger workspace register $\regW'$; 
when we analyze the state before the final measurement, we group $\regI$, $\regO$, and $\regW$ together in a larger workspace register $\regW''$; see Figure~\ref{fig:algoCircuit}.
We may also group together all registers of the algorithm, $\regI$, $\regO$, $\regW$, and $\regR$, in a single register $\regA$.

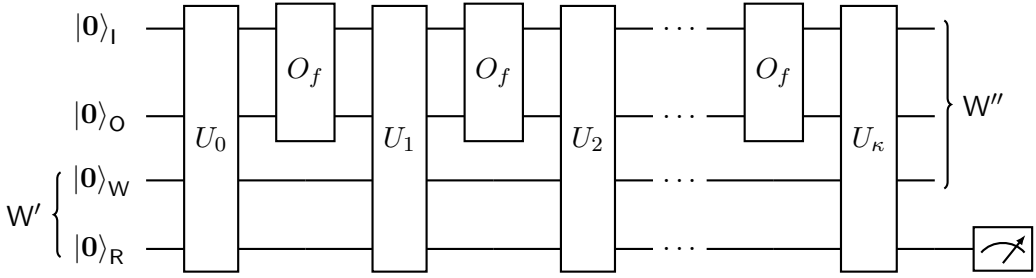
\begin{figure}[ht!]
	\centering
\begin{tikzcd}
&&
\lstick{$\ket{\mathbf{0}}_{\regI}$\hspace{5.9pt}}
& \gate[wires=4]{U_0} & \gate[wires=2]{O_f} & \gate[wires=4]{U_1} & \gate[wires=2]{O_f}& \gate[wires=4]{U_2} &\ \ldots\ \qw & \gate[wires=2]{O_f} & \gate[wires=4]{U_\kappa} & \qw \rstick[wires=3]{$\regW''$}
\\ 
&&\lstick{$\ket{\mathbf{0}}_{\regO}\hspace{2.5pt}$} 
& & & & & &\ \ldots\ \qw & & & \qw 
\\ 
\lstick[wires=2]{$\regW'$}&&
 \lstick{$\ket{\mathbf{0}}_{\regW}$\hspace{0.9pt}} 
& & \qw & & \qw & &\ \ldots\ \qw & \qw & & \qw
\\  
&&
\lstick{$\ket{\mathbf{0}}_{\regR}$\hspace{3.2pt}} 
&& \qw & & \qw & &\ \ldots\ \qw & \qw &  & \qw & \meter{}
\end{tikzcd}
\caption{\small The circuit diagram of a generic quantum query algorithm.}
\label{fig:algoCircuit}
\end{figure}

The algorithm starts the computation in a state that is independent from the function, and, without loss of generality, let all its qubits be initialized to $0$.
Let $|\mathbf{0}\>$ denote the state of any number of qubits that are all set to zero.
Given this, the algorithm is completely characterized by function-independent unitaries $U_0,U_1,\ldots,U_\kappa$ acting on all registers.%
\footnote{
One may also assume that all $U_0,U_1,\ldots,U_\kappa$ are equal because the workspace register can keep a clock subregister that controls which action to perform.
}
Then the algorithm alternates between applying unitary $U_k$ and query $O_f\otimes I_{\regW'}$,
where $k$ takes values $0,1,2,\ldots,\kappa-1$ in that order and where $I_{\mathsf{reg}}$ denotes the identity operator acting on a register $\mathsf{reg}$.
Finally, the algorithm applies the unitary $U_\kappa$, measures the result register (w.l.o.g., in the standard basis), and returns the measurement result.
We may also sometimes ignore the final measurement, still referring to such quantum process as an algorithm.
The number $\kappa$ is the total number of oracle calls, and it is called the query complexity of the algorithm.

We say that the algorithm succeeds if the result it returns is a correct solution for the problem of our interest on the function $f$.
Both for the Inverting Permutations and Unstructured Search problems, the correct result is $x\in\Dom$ such that $f(x)=\zero$.
The probability with which the algorithm returns the correct answer is called the \emph{success probability},  $p_{succ}(f)$, and it may depend on $f$.

We will consider the case when $f$ is chosen according to some probability distribution $(\delta_f)_{f\in F}$ from some set $F$.
In that case, we may call $p_{succ}=\sum_{f\in F}\delta_f p_{succ}(f)$ the \emph{overall} success probability.
The overall success probability depends on the probabilistic nature of quantum computing plus the random choice of $f\in F$.

\subsection{Introducing Function Register}
\label{sec:FuncReg}

For purposes of proving hardness, we may think of the algorithm as a player that, without initially knowing anything about $f$ aside from the fact that it belongs to $F$, interacts with another player that completely knows the function $f$.
We simply refer to this other player as the \emph{oracle}.%
\footnote{
In the adversary method, this player is commonly called the adversary, because its task is to thwart the progress of the algorithm.
Here, however, we avoid using the name ``adversary'', as in the cryptographic settings, the adversary refers to the other player, that is, the algorithm itself.
}
Thus we use the word oracle to mean both the oracle call, namely, the unitary transformation $O_f$, as well as the player in the interaction; the meaning should be clear from the context.

In the complete quantum system of the two players, let us denote the register held by the oracle as $\regF$, and we may refer to it both as the \emph{function register} and as the \emph{oracle register}.
This register corresponds to the complex Euclidean space $\mathcal{F}:=\mathbb{C}^{F}$, with an orthonormal basis $\{|f\>_\regF\colon f\in F\}$.

In this algorithm-oracle interaction, we would like that, if $\regF$ is initialized to $|f\>$, the algorithm performs exactly as if it had access to the oracle call $O_f$.
Furthermore, if $\regF$ is initialized to a superposition $\sum_{f\in F}\sqrt{\delta_f}\ket{f}_{\regF}$,
we would like that the algorithm acts as if it was given a random $f\in F$ with probability $\delta_f$. 
For that reason, let us first define the \emph{extended} oracle acting on registers $\regF,\regI,\regO$ as
\[
O:=\sum_{f\in F}|f\>\<f|_\regF\otimes(O_f)_{\regI\regO}.
\]
We may also still refer to the extended oracle simply as the oracle. 
Similarly to the original scenario, where the computation alternated between applications of $U_k$ and $O_f\otimes I_{\regW'}$, now the computation alternates between applications of $I_\regF\otimes U_k$ and $O\otimes I_{\regW'}$; see Figure~\ref{fig:interaction}.
Finally, both registers $\regF$ and $\regR$ are measured, obtaining some function $f$ and some result $r$, and we say that the algorithm-oracle interaction succeeded if $r$ is a correct solution on the function $f$. 
Note that $\regF$ measures to $|f\>$ with probability $\delta_f$. 
We sometimes refer to the combined algorithm-oracle system as the \emph{overall} system.

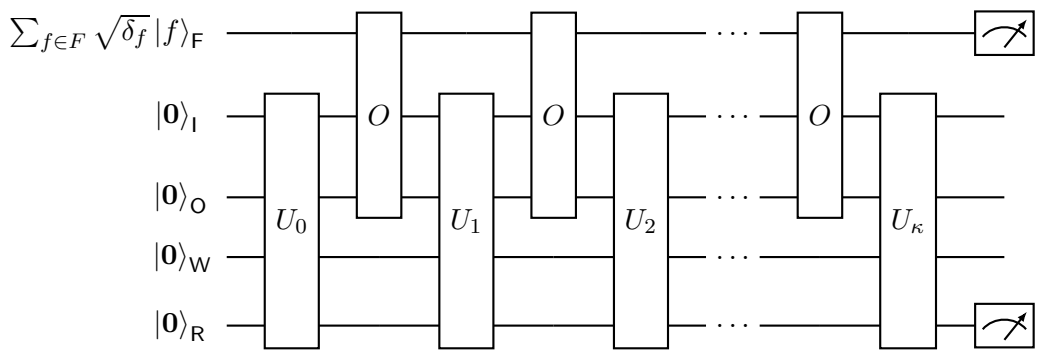
\begin{figure}[ht!]
	\centering
\begin{tikzcd}
 \lstick{$\sum_{f\in F}\sqrt{\delta_f}\ket{f}_{\regF}\hspace{3.5pt}$}
  & \qw & \gate[wires=3]{O} &\qw & \gate[wires=3]{O} & \qw &\ \ldots\ \qw & \gate[wires=3]{O} & \qw
  & \meter{}
\\ 
 \lstick{$\ket{\mathbf{0}}_{\regI}$\hspace{5.9pt}}
& \gate[wires=4]{U_0} & & \gate[wires=4]{U_1} & & \gate[wires=4]{U_2} &\ \ldots\ \qw & & \gate[wires=4]{U_\kappa} & \qw 
\\ 
\lstick{$\ket{\mathbf{0}}_{\regO}\hspace{2.5pt}$} 
& & & & & &\ \ldots\ \qw & & & \qw 
\\ 
 \lstick{$\ket{\mathbf{0}}_{\regW}$\hspace{0.9pt}} 
& & \qw & & \qw & &\ \ldots\ \qw & \qw & & \qw
\\  
 \lstick{$\ket{\mathbf{0}}_{\regR}$\hspace{3.2pt}} 
& & \qw & & \qw & &\ \ldots\ \qw & \qw &  
 & \meter{}
\end{tikzcd}
\caption{\small The circuit diagram after the function register has been introduced.}
\label{fig:interaction}
\end{figure}

\subsection{Unstructured Search and Inverting Permutations}

In the Unstructured Search problem, $\Func$ is the set of all $M^N$ functions $f\colon\Dom\rightarrow\Ran$. In the Inverting Permutations problem, $M=N$, $\Dom=\Ran$, and $\Func$ is the set of all $N!$ permutations $f\colon\Dom\rightarrow\Ran$. For both problems, we are given an access to one $f\in\Func$ chosen uniformly at random, and the task of the algorithm is to find $x$ such that $f(x)=\zero$. Note that, while for Inverting Permutations there is exactly one correct solution, for Unstructured Search, depending on $f$, there can be any number of solutions, including none.

 Many statements that we make apply for both problems. When we want to address a particular problem, for sake of brevity, we may call the Inverting Permutations problem the \emph{permutation case} and the Unstructured Search problem the \emph{function case}. Or, even more briefly, we may say that some statement holds \emph{for permutations} or \emph{for functions} (e.g., see Theorem~\ref{thm:mainThm}).
 
Since we are given each function $f$ with equal probability $\delta_f=1/|\Func|$,
the function register is initialized to
\[
|\vn_\emptyset\>:=\frac{1}{|\Func|}\sum_{f\in\Func}|f\>.
\]
The choice for this notation will become clear in Section~\ref{sec:dbSubspaces}. For now, let us mention that $\emptyset$ in $|\vn_\emptyset\>$ indicates that initially the algorithm has not yet learnt anything about the function it is interacting with.

\begin{defn}
\label{def:PhiPsi}
Define $|\phi_0\>:=\ket{\vn_\emptyset}\otimes|\mathbf{0}\>_\regA$ and, recursively,  $|\psi_{k}\>:=(I_\regF\otimes U_{k})|\phi_{k}\>$ for $k=0,1,\ldots,\kappa$ and $|\phi_{k}\>:=(O\otimes I_{\regW'})|\psi_{k-1}\>$ for $k=1,\ldots,\kappa$.
\end{defn}

Hence, $|\phi_0\>$ is the initial state of the overall system, $|\psi_\kappa\>$ is the final state of the system, that is, the state before the measurement, and, for the other values of $k$, $|\psi_{k-1}\>$ and $|\phi_{k}\>$ are, respectively, the states right before and right after $k$-th oracle call.

For every pair $(x,y)\in\Dom\times\Ran$, let 
\[
\Xi_x^y := \sum_{f\colon f(x)=y}|f\>\<f|
\]
be the orthogonal projector acting on the space $\mathcal{F}$ that projects on the subspace corresponding to all $f\in\Func$ such that $f(x)=y$. Note that, for every $x$, we have $\sum_{y\in\Ran} \Xi_x^y=I_\regF$.%
\footnote{For permutations, we also have $\sum_{x\in\Dom}\Xi_x^y=I_\regF$ for every $y$. This, however, does not hold for functions.}

It is convenient to rewrite the extended oracle call as
\begin{equation}
\label{eq:ODecomp}
O = \sum_{x,y,y'}\Xi^y_x \otimes |x\>\<x|\otimes |y'\oplus y\>\<y'|.
\end{equation}
For both problems, the result register is $N$-dimensional. Thus, the orthogonal projector that projects on all correct function-result pairs can be written as 
\begin{equation}
\label{eq:PDecomp}
P:=\sum_{x\in\Dom}\Xi^\zero_x\otimes I_{\regW''}\otimes|x\>\<x|_\regR.
\end{equation}
In this notation, the overall success probability is $\|P|\psi_\kappa\>\|^2$.

\subsection{Additional preliminaries}

The \emph{support} of an Hermitian operator is the space spanned by its eigenvectors corresponding to non-zero eigenvalues. We say that a Hermitian operator is \emph{supported} on a space $\cH$ if its support is a subspace of $\cH$ (or the space $\cH$ itself).

Throughout the paper, we follow the convention that, if a subspace (mostly of the function space $\cF$) is denoted by $\cH$  with some subscripts, superscripts, and overscripts, then $\Pi$ with these subscripts, superscripts, and overscripts denote the orthogonal projector on that subspace.
For example, $\Pi^\hi$ and $\overline\Pi^\lo_k$ denote orthogonal projectors on $\mathcal{H}^\hi$ and $\overline{\mathcal{H}}^\lo_k$, respectively (we introduce these subspaces of $\cF$ in Section~\ref{sec:dbSubspaces}).

Suppose $\mathcal{H}$ is a subspace of $\mathcal{F}$ and $\Pi$ is the orthogonal projector on $\cH$.
Given a state $|\zeta\>$ of the overall system, we say that $|\zeta\>$ has \emph{amplitude} $\|(\Pi\otimes I_\regA)|\zeta\>\|$ on $\mathcal{H}$ and \emph{weight}
\begin{equation}
\label{eqn:weightOnH}
\big\|(\Pi\otimes I_\regA)|\zeta\>\big\|^2
=
\tr\Big[\Pi\tr_\regA\big[|\zeta\>\<\zeta|\big]\Big]
\end{equation}
 on $\mathcal{H}$.
 Similarly, given a state $|u\>\in\cF$, we say that $|u\>$ has amplitude $\|\Pi|u\>\|$ and weight $\|\Pi|u\>\|^2$ on $\cH$.
 Although these two quantities, amplitude and weight, are quadratically related, on one hand side, it may be more natural to think in terms of weight when we want to think of probability,
 while, on the other hand, thinking of amplitudes may be more convenient when we want to see how these quantities change upon calls to the extended oracle.
 We may also informally refer to both the weight and the amplitude as the \emph{overlap}.

Given an operator $\Phi$ that acts on some subset of the registers of the overall system, let $\widetilde\Phi$ denote its extension to the whole overall system where it acts as the identity on the other registers. For example,
$\widetilde\Pi^\hi := \Pi^\hi \otimes I_\regA$ and $\widetilde{O}:=O\otimes I_{\regW'}$.

In the paper, we use several properties of the spectral norm, mostly without explicitly referring to them.
Let us list some of them here. Note that Property~5 is a direct consequence of Property~4.
\begin{fact}
\label{fac:norms}
For matrices $\Phi$ and $\Psi$ of appropriate dimensions,
\begin{enumerate}
\item $\|\Phi\|=\|\Phi\Phi^*\|^{1/2}$;
\item $\|\Phi\Psi\|\le\|\Phi\|\|\Psi\|$;
\item if $\Phi$ and $\Psi$ have orthogonal images, then $\|\Phi\|\le \|\Phi+\Psi\|$;
\item if $\Phi$ and $\Psi$ have orthogonal images and coimages, then $\|\Phi+\Psi\|=\max\{\|\Phi\|,\|\Psi\|\}$;
\item if $\Phi$ is square and $\{\Pi_i\}_i$ is a collection of orthogonal projectors such that $\Phi$  and $\Pi_i$ commute for all $i$ and $\sum_i\Pi_i=I$, then $\|\Phi\|=\max_i\|\Phi\Pi_i\|$.
\end{enumerate}
\end{fact}

\section{Oracle Subspaces: Interpretation and Application}
\label{sec:dbSubspaces}

Recall that the function register $\regF$ corresponds to $|\Func|$-dimensional complex Euclidean space $\mathcal{F}$, with $\{|f\>\colon f\in\Func\}$ as the standard orthonormal basis. In this section, we consider various subspaces of $\cF$.

Throughout the execution of the algorithm, the reduced density operator of the oracle evolves, and its overlap on certain subspaces allows us to measure the progress of the algorithm. More precisely, we will decompose $\cF$ as a direct sum of subspaces $\cH^\hi\subset \cF$ and $\cH^\lo\subset \cF$ that essentially correspond to the situation, respectively, where the algorithm has or has not learnt a preimage of $\zero$. We further decompose $\cH^\hi$ and $\cH^\lo$ as direct sums of subspaces $\overline\cH_k^\hi\subset\cH^\hi$ and $\overline\cH_k^\lo\subset \cH^\lo$ that essentially  correspond to the situation when the algorithm has learnt the output of the function on $k$ inputs. The superscirpts ``high'' and ``low'' hint to the success probability of the algorithm. 

We start in Section~\ref{sec:hlssDef} with formal definitions of all the relevant spaces. Then, in Section~\ref{ssec:mainThm}, we give a brief intuition on how to interpret the situation where one has large weigh to those spaces and we state our main technical theorem, which aligns with that interpretation. The corollary of the theorem is the tight upper bound on the success probability. Then, once we know that our interpretation of spaces is useful, in Section~\ref{ssec:solFound} we provide further backing to our interpretation by discussing certain simple algorithms, in particular focusing on the permutation case.

\subsection{Definitions of High and Low Success Subspaces}
\label{sec:hlssDef}

Given $k$ distinct elements $x_1,\ldots,x_k\in\Dom$, we call a function
\[
\alpha\colon\{x_1,\ldots,x_k\}\rightarrow\Ran\colon x_i\mapsto y_i
\]
an \emph{assignment} of weight $|\alpha|:=k$. In the permutation case, we only consider injective assignments. We denote the ``empty assignment'' of weight $0$ by $\emptyset$. We use both $|v_\alpha\>$ and $|\vn_{x_1,\ldots,x_k}^{y_1,\ldots,y_k}\>$ to denote the uniform superposition over all functions $f$ that agree with $\alpha$ on all $x\in\mathrm{dom}\,\alpha$, that is, $f(x_i)=y_i=\alpha(x_i)$ for all $i\in\{1,\ldots,k\}$.%
\footnote{This vector in its essence is similar to Zhandry's database state $|\{(x_1,y_1),\ldots,(x_n,y_n)\}\>$ in \cite{zhandry:record} and it equals Unruh's $|\alpha\>^{\mathfrak{f}}$ and $|\alpha\>^{\mathfrak{p}}$ in \cite{unruh:perm}.}
 Note that the unit vector $|v_\alpha\>$ is a superposition is over $N^{M-k}$ functions in the function case and over $(N-k)!$ permutations in the permutation case.
 
We define
\[
\spA{k}:=\spn\big\{|\vn_\alpha\>\colon |\alpha|=k \big\}
\qquad\text{and}\qquad
\spB{k}:=\spn\{|\vn_\alpha\>\colon |\alpha|=k \And \zero \in \mathrm{im}\,\alpha \big\}.
\]
We clearly have $\spB{k}\subseteq\spA{k}$.
Also note that
\[
 |\vn_{x_1,\ldots,x_{k-1}}^{y_1,\ldots,y_{k-1}}\>=\sqrt{N}
 \sum_{y'\in\Ran}
 |\vn_{x_1,\ldots,x_{k-1},x'}^{y_1,\ldots,y_{k-1},y'}\>
\]
for functions and
\[
 |\vn_{x_1,\ldots,x_{k-1}}^{y_1,\ldots,y_{k-1}}\>=\sqrt{N-k}
 \hspace{-15pt}
 \sum_{y'\in\Ran\setminus \{y_1,\ldots,y_{k-1}\}}
 \hspace{-20pt}
 |\vn_{x_1,\ldots,x_{k-1},x'}^{y_1,\ldots,y_{k-1},y'}\>
 =\sqrt{N-k}
 \hspace{-15pt}
 \sum_{x''\in\Dom\setminus \{x_1,\ldots,x_{k-1}\}}
 \hspace{-20pt}
 |\vn_{x_1,\ldots,x_{k-1},x''}^{y_1,\ldots,y_{k-1},y''}\>
\]
for permutations, where $x'$ and $y''$ are arbitrary elements of $\Dom\setminus \{x_1,\ldots,x_{k-1}\}$ and $\Ran\setminus \{y_1,\ldots,y_{k-1}\}$, respectively. Therefore we have $\spA{k-1}\subseteq\spA{k}$ and, in the permutation case, by taking $y''=\zero$, we also have $\spA{k-1}\subseteq\spB{k}$.
Define 
\[
\overline\cH_k^\hi := \spB{k}\cap(\spA{k-1})^\perp
\qquad\text{and}\qquad
\overline\cH_k^\lo := \spA{k}\cap(\spB{k}+\spA{k-1})^\perp,
\]
where $\cH^\perp$ is the orthogonal complement of $\cH$ in $\cF$ and the $+\spA{k-1}$ term in the definition of $\overline{\cH}_k^\lo$ is redundant in the permutation case.
Note that the spaces 
\[
\overline\cH_0^\lo, \overline\cH_1^\hi, \overline\cH_1^\lo, \overline\cH_2^\hi, \overline\cH_2^\lo, \ldots
\]
are mutually orthogonal, and, as we will implicitly show, they are also all non-trivial up to $\overline\cH_{M}^\lo$ for functions and $\overline\cH_{N-1}^\hi$ for permutations.%
\footnote{See Claim~\ref{clm:spABviaReps} in Section~\ref{sec:spAspBviaReps} for the permutation case.}
Let us define
\begin{equation}
\label{eqn:cHaccum}
\cH^\lo_k := 
\bigoplus_{k'=0}^k
\overline\cH^\lo_{k'}
\qquad\text{and}\qquad
\cH^\hi_k := 
\bigoplus_{k'=1}^k
\overline\cH^\hi_{k'},
\end{equation}
and note that $\cH^\lo_k\oplus\cH^\hi_k=\spA{k}$.
In the function case, when considering connections to Zhandry's approach, we will also consider $\cH^\lo:=\cH^\lo_M$ and $\cH^\hi:=\cH^\hi_M$, for which we have $\cH^\lo\oplus\cH^\hi=\cF$.

\subsection{Main Theorem}
\label{ssec:mainThm}

Recall (\ref{eqn:weightOnH}), the definition of the weight of the overall state of the system on a subspace of $\cF$. 
As before, superscripts ``$\lo$'' and ``$\hi$'' refer to success probability, and we call $\cH^\lo_k$ and $\cH^\hi_k$ \emph{low} and \emph{high} subspaces, respectively.
Intuitively, the weight being on the space $\cH^\lo_k\oplus\cH^\hi_k$ means that the algorithm has learnt the outputs of the function on \emph{at most} $k$ inputs, and the weight being on $\cH^\hi_k$ means that one of those outputs is $\zero$,
while the weight being on the space $\overline\cH^\lo_k\oplus\overline\cH^\hi_k$ means that the algorithm has learnt the outputs of the function on \emph{exactly} $k$ inputs.
More formally, Part 1 of Theorem~\ref{thm:mainThm} states that, if an algorithm has performed $k$ queries, the state of the function register is supported on $\cH^\lo_k\oplus \cH^\hi_{k-1}$.
In Part 2 of Theorem~\ref{thm:mainThm} we formally show that the weight on $\cH^\lo_k$ contribute little to the success probability.

For spaces $\cH_k^\labB$ and $\cH_k^\labG$, Theorem~\ref{thm:mainThm} states that, on the one hand side, initially all the weight is on the low subspace $\cH_k^\labB$. On the other hand, using calls to the oracle, one has to try to transfer this weight to the high subspace $\cH_k^\labG$ because, in the end, 
the overall success probability is approximately no more than $1/N$ plus the weight of the final state $|\psi_\kappa\>$ on $\cH_\kappa^\labG$.
 However, Theorem~\ref{thm:mainThm} also limits how much weight can be transferred per a single oracle call: the amplitude on the high subspace $\cH_k^\labG$ cannot increase by more than $\OO(1/\sqrt{N})$.

Recall that $\Pi_k^\lo$ and $\Pi_k^\hi$ are projectors on $\cH_k^\lo$ and $\cH_k^\hi$ respectively, and 
\[
\widetilde\Pi_k^\lo = \Pi_k^\lo \otimes I_\regA
\qquad\text{and}\qquad
\widetilde\Pi_k^\hi = \Pi_k^\hi \otimes I_\regA.
\]

\begin{thm}
\label{thm:mainThm}
Let $|\phi_{k}\>$ and $|\psi_{k}\>$ be as in Definition~\ref{def:PhiPsi}.
Let $\alpha_0:=0$ and 
\[
\alpha_k
:=\big\|\widetilde\Pi^\hi_k|\phi_k\>\big\|
=\big\|\widetilde\Pi^\hi_k|\psi_k\>\big\|
\]
for $k>0$,
namely, the amplitude of the state of the system on $\cH^\hi_k$ after $k$ oracle calls.
We have
\begin{enumerate}
\item $(\widetilde\Pi^\lo_k+\widetilde\Pi^\hi_k)|\phi_k\>=|\phi_k\>$ and
$(\widetilde\Pi^\lo_k+\widetilde\Pi^\hi_k)|\psi_k\>=|\psi_k\>$;
\item $\|P|\psi_k\>\|\le \alpha_k+
\begin{cases}
\frac{1}{\sqrt{N}} & \text{for functions}, \\
\frac{1}{\sqrt{N-2k}} & \text{for permutations};
\end{cases}$
\item $\alpha_k\le \alpha_{k-1}+
\begin{cases}
\frac{\sqrt{2}}{\sqrt{N-1}} & \text{for functions}, \\
\frac{2\sqrt{2}}{\sqrt{N-4k}} & \text{for permutations}.
\end{cases}$
\end{enumerate}
\end{thm}

As a consequence, $\alpha_0=0$ and Point 3 of the theorem imply that $\alpha_\kappa\le \sqrt{2}\kappa/\sqrt{N}$ for functions and $\alpha_\kappa\le {2\sqrt{2}\kappa}/\sqrt{N-4\kappa}$ for permutations. That, combined with Point 2 gives us the asymptotically optimal bound on the success probability.

\begin{cor}
Suppose a quantum algorithm is given an oracle access to a uniformly random function $f\colon\Dom\rightarrow\Ran$, where $|\Ran|=N$. If the algorithm makes $\kappa$ oracle calls, the probability of it solving the Unstructured Search problem, namely, finding a preimage of $\zero$, is no more than
$(1+\sqrt{2}\kappa)^2/(N-1)$.
\end{cor}

\begin{cor}
Suppose a quantum algorithm is given an oracle access to a uniformly random permutation $f\colon\Dom\rightarrow\Dom$, where $|\Dom|=N$. If the algorithm makes $\kappa$ oracle calls, the probability of it solving the Inverting Permutation problem, namely, evaluating $f^{-1}$ on $\zero$, is no more than
$(1+2\sqrt{2}\kappa)^2/(N-4\kappa)$.
\end{cor}

Notice that Point 1 of Theorem~\ref{thm:mainThm} is not explicitly used to prove the above bounds. In the function case, Point 1 is not required at all, as we will see in Section~\ref{sec:ZhandryApproach} when discussing connection to Zhandry's approach. In the permutation case, however, Point 1 used to prove Points 2 and 3 of the theorem. Because of Point 1, instead of defining $\alpha_k$ as $\big\|\widetilde\Pi^\hi_\ell|\phi_k\>\big\|$ for some large $\ell$ to encapsulate the meaning that $\alpha_k$ depicts the amplitude of $|\phi_k\>$ on high success subspaces, it suffices to define $\alpha_k$ as $\big\|\widetilde\Pi^\hi_k|\phi_k\>\big\|$.
Intuitively, unlike for functions, for permutations, when we make a large number of queries, even if we have not found a preimage of $\zero$, our chances of randomly guessing that preimage as well as our chances of discovering it by the next query increase.
Hence, it is important that $|\phi_k\>$ has no overlap on $\overline\cH_\ell^\hi$ with $\ell >k$, especially, $\ell \gg k$.

Consider the Inverting Permutation problem, where every $f$ has exactly one valid solution. A single query to Grover's search can be implemented using $2$ oracle calls.
Hence, using an even number $\kappa$ of queries, Grover's algorithm finds $f^{-1}(\zero)$ with probability
\[
\sin^2\big((1+\kappa)\arcsin(1/\sqrt{N})\big),
\]
which is approximately $(1+\kappa)^2/N$ when $0<\kappa\ll N$.
Thus our bound is tight up to the factor of about $8$.

We provide a framework for proving Thorem~\ref{thm:mainThm} in Section~\ref{sec:MainTmFrame}. We prove the function case in Section~\ref{sec:ZhandryApproach}. The proof for the permutation case requires the representation theory of the symmetric group.
In Section~\ref{sec:RepPrelim} we introduce various tools from the theory required by our proofs.
In Section~\ref{sec:ABviaRepTh} we express $\overline\cH^\hi_k$ and $\overline\cH^\lo_k$ as direct sums of certain irreducible representations.
In Section~\ref{sec:MainTmProof} we prove Theorem~\ref{thm:mainThm}.

\subsection{Additional observations on $\spA{k}$ and $\spB{k}$}
\label{ssec:solFound}

To better understand the space $\spA{k}$, consider a simple algorithm that initially creates a uniform superpositions over all $k$-tuples $(x_1,\ldots,x_k)$ of distinct elements in $\Dom$ and then computes $|f(x_1),\ldots,f(x_k)\>$ in a separate register using $k$ oracle calls. It is not hard to see that the resulting reduced state of the oracle is a uniform mixture over all $|\vn_\alpha\>\<\vn_\alpha|$ such that $|\alpha|=k$, and $\spA{k}$ is the support of this state. Further on, Point 1 of Theorem~\ref{thm:mainThm} shows that, for any algorithm that makes $k$ oracle calls, the reduced state of the oracle is supported on $\spA{k}$. Hence, $\spA{k}$ is exactly the subspace of $\cF$ that can be ``reached'' by interacting with quantum algorithms that make at most $k$ oracle calls.

As mentioned before, we interpret the weight being on the space $\spA{k}=\cH^\lo_k\oplus\cH^\hi_k$ as the algorithm having learnt the outputs of the function on \emph{at most} $k$ inputs.
The phrase ``at most'' is important here, because $\spA{k}$ contains $\spA{k-1}$. Therefore, intuitively, to say that the algorithm has learnt the value of the function on \emph{exactly} $k$ inputs, the weight must be on $\spA{k}\cap(\spA{k-1})^\perp=\overline\cH^\lo_k\oplus\overline\cH^\hi_k$. 
While this intuition is not completely accurate, because a part of the state $  |\vn_{x_1,\ldots,x_k}^{y_1,\ldots,y_k}\>$ overlaps $\spA{k-1}$, we can back it up by the following claim, which we prove in Section~\ref{sec:ZhandryApproach} for functions and in Appendix~\ref{app:LargeOlaps} for permutations.

\begin{clm}
\label{clm:LargeOlap0}
The norm of the projection of $  |\vn_{x_1,\ldots,x_k}^{y_1,\ldots,y_k}\>$ on the space $\spA{k}\cap(\spA{k-1})^\perp$ is
\begin{itemize}
\item $\sqrt{(1-1/N)^k}$ for functions,
\item between $\sqrt{1-k/(N-k+1)}$ and $\sqrt{1-k/N}$ for permutations.
\end{itemize}
\end{clm}

Now let provide some further observations about the space $\spB{k}$. In the reminder of the section let us only consider the permutation case, because some of the arguments become much easier by the fact that every permutation has exactly one correct solution. We will provide some intuition why we define spaces $\overline\cH^\hi_k$ and $\overline\cH^\lo_k$ the way we do, and what other intuitively promising definitions might be misleading.

\paragraph{Subspaces corresponding to solution being found.}

Consider a scenario where the algorithm has perfectly computed $|f^{-1}(\zero)\>$ and nothing else.
It could be achieved in various ways. For example, by running the exact Grover's search or, more slowly, by computing $f(x)$ on all inputs $x$, writing $x'$ such that $f(x')=\zero$ in the result register, and then uncomputing all $f(x)$.
Either way, the reduction of the final state to the function register $\regF$ is $\sum_{x\in\Dom}|\vn_x\>\<\vn_x|/N$, 
which is the maximally mixed state on $\spB{1}$.
Similarly, we can consider a somewhat contrived scenario where, after learning $|f^{-1}(\zero)\>$, the algorithm prepares the uniform superposition over $(k-1)$-tuples of distinct elements in $\Dom\setminus\{f^{-1}(\zero)\}$ and queries $f$ on all $k-1$ inputs.
In this case, the reduced state of the oracle is the uniform mixture over states $|\vn_{x_1,x_2,\ldots,x_{k}}^{\zero,y_2,\ldots,y_{k}}\>$, and its support is 
$\spB{k}$.

Thus, it might be tempting to interpret the oracle's state overlapping $\spB{k}$ as the algorithm having learnt the value of the function on at most $k$ inputs and, in addition, for one of them this value being $\zero$. However, such interpretation is problematic because $\spA{k-1}\subseteq\spB{k}$ in the permutation case and, even in the function case, states in $\spA{k-1}$ may have a significant overlap on $\spB{k}$.
Therefore, for the purposes of such an interpretation, we choose to consider subspace $\overline\cH_k^\hi = \spB{k}\cap(\spA{k-1})^\perp$ instead of $\spB{k}$.
Note that, since we have $\spB{k}\subseteq\spA{k}$, Claim~\ref{clm:LargeOlap0} states that the overlap of $|\vn_{x_1,x_2,\ldots,x_{k}}^{\zero,y_2,\ldots,y_{k}}\>$ on $\overline\cH_k^\hi$ is $\sqrt{1-\OO(k/n)}$.

\paragraph{Subspaces corresponding to solution not being found.}

We have chosen to define the subspace that (approximately) corresponds to the algorithm having learnt the value of the function on exactly $k$ inputs none of whom is the preimage of $\zero$ as $\overline\cH_k^\lo = \spA{k}\cap(\spB{k})^\perp$.
Another natural idea for defining such a subspace would be to define it as 
\[
\spA{k}^{\mathrm{alt}}:=\spn\big\{|\vn_{x_1,\ldots,x_{k}}^{y_1,\ldots,y_{k}}\>
\colon x_1,\ldots,x_{k}\in\Dom\And y_1,\ldots,y_{k}\in\Ran\setminus\{\zero\} \big\}.
\] 
However, \cite{unruh:perm} reports the observation by Minki Hhan that this interpretation of $\spA{k}^{\mathrm{alt}}$ is misleading.
Indeed, in the permutation case, we have $\spA{k}^{\mathrm{alt}}=
\spA{k}$, as seen from
\begin{align*}
 |\vn_{x_1,\ldots,x_{k-1},x_{k}}^{y_1,\ldots,y_{k-1},\zero}\>
 & =\frac{1}{ \sqrt{N-k}}
 |\vn_{x_1,\ldots,x_{k-1}}^{y_1,\ldots,y_{k-1}}\>
 -
 \hspace{-15pt}
 \sum_{y_k\in\Ran\setminus \{\zero,y_1,\ldots,y_{k-1}\}}
 \hspace{-30pt}
 |v_{x_1,\ldots,x_{k-1},x_{k}}^{y_1,\ldots,y_{k-1},y_k}\>
 \\ & =
 \hspace{-10pt}
 \sum_{x'\in\Dom\setminus \{x_1,\ldots,x_{k-1}\}}
 \hspace{-20pt}
 |\vn_{x_1,\ldots,x_{k-1},x'}^{y_1,\ldots,y_{k-1},y'}\>
 -
 \hspace{-15pt}
 \sum_{y_k\in\Ran\setminus \{\zero,y_1,\ldots,y_{k-1}\}}
 \hspace{-30pt}
 |\vn_{x_1,\ldots,x_{k-1},x_{k}}^{y_1,\ldots,y_{k-1},y_k}\>
 \in \spA{k}^{\mathrm{alt}}, 
\end{align*}
where $y'$ is an arbitrary value in $\Ran\setminus \{\zero,y_1,\ldots,y_{k-1}\}$,
and thus $\spB{k}\subseteq\spA{k}^{\mathrm{alt}}$.
Hence, saying that the support of a state lays within $\spA{k}^{\mathrm{alt}}$ does not exclude the scenario of it laying within $\spB{k}$.

The problem is that, while each state spanning $\spA{k}^{\mathrm{alt}}$, i.e., $|\vn_{x_1,\ldots,x_{k}}^{y_1,\ldots,y_{k}}\>$ has only a minuscule overlap on $\spB{k}$, as quantified by the claim below, their linear combination may lay within $\spB{k}$.

\begin{clm}
\label{clm:LargeOlap1}
Given $y_1,\ldots,y_k\ne \zero$, the projection of $|\vn_{x_1,\ldots,x_{k}}^{y_1,\ldots,y_{k}}\>$ on $\overline\cH_k^\lo$ has the norm between $\sqrt{1-k/(N-k)}$ and $\sqrt{1-k/(N-1)}$.
\end{clm}

We prove Claim~\ref{clm:LargeOlap1} along the same lines as Claim~\ref{clm:LargeOlap0} in Appendix~\ref{app:LargeOlaps}.
These two claims are not used when proving the main results of the paper, but they give valuable intuition about the relation between states $|\vn_{x_1,\ldots,x_{k}}^{y_1,\ldots,y_{k}}\>$ and $|\vn_{x_1,x_2,\ldots,x_{k}}^{\zero,y_2,\ldots,y_{k}}\>$ with  $y_1,\ldots,y_k\ne \zero$ and spaces $\overline\cH_k^\lo$ and $\overline\cH_k^\hi$.

\section{Framework for the Proof of the Main Theorem}
\label{sec:MainTmFrame}

In this section, we provide framework for proving Theorem~\ref{thm:mainThm}. We start with a full proof for Point~1 of the theorem, this being the simplest point. Then we show how the proof of Point~2 can be reduced to upper-bounding the norm of $\Xi_x^\zero\Pi^\lo_{k}$. The proof of Point~3 can also be reduced to upper-bounding norms of certain operators. In particular, we need to bound the norm given in Eqn.~(\ref{eq:NormSimpl}). While, for the function case, we do know how to bound this norm, for the permutation case, we have to go further and we can bound it by bounding the the norms of the operators given in Eqn.~(\ref{eqn:GBTriangle}). 

As for the actual bounds for these norms, for the function case, we provide them in Section~\ref{sec:ZhandryApproach}, where we discuss connections to Zhandry's compressed oracle model. For the permutation case, we bound them in Section~\ref{sec:MainTmProof}, after we have introduced the representation theory of symmetric group in Section~\ref{sec:RepPrelim} and expressed the relevant operators in its terms in Section~\ref{sec:ABviaRepTh}.

\paragraph{Spaces reachable by $k$ queries (proof of Point 1).}

We prove Point 1 by induction on $k$. Note that for all $k$ the operator $\Pi^\lo_k+\Pi^\hi_k$ is the orthogonal projector on $\cH^\lo_k\oplus\cH^\hi_k=\spA{k}$. 

For the base case $k=0$, both $|\phi_0\>$ and $|\psi_0\>$ are product states with respect to the oracle-algorithm separation, and their reduces state of the oracle part is $|\vn_\emptyset\>\in\spA{0}$.

For the inductive step, recall the expression~(\ref{eq:ODecomp}) that expresses $O$ via projectors $\Xi_x^y$. It suffices to show that, when $\Xi_x^y$ acts on the vectors spanning $\spA{k}$, namely, vectors in the form $|\vn_{x_1,\ldots,x_{k}}^{y_1,\ldots,y_{k}}\>$, the result lays in $\spA{k+1}$.
In the permutation case, we have
\[
\Xi_x^y |\vn_{x_1,\ldots,x_{k}}^{y_1,\ldots,y_{k}}\> =
\begin{cases}
|\vn_{x_1,\ldots,x_{k}}^{y_1,\ldots,y_{k}}\> & \text{if }\exists i\colon (x_i,y_i)=(x,y),\\
\frac{1}{\sqrt{N-k}}|\vn_{x_1,\ldots,x_{k},x}^{y_1,\ldots,y_{k},y}\> & \text{if }\forall i\colon x_i\ne x\text{ and }\forall i\colon y_i\ne y, \\
0 & \text{if }\exists i\colon (x_i= x\And y_i\ne y) \text{ or } i\colon (x_i\ne x\And y_i = y),
\end{cases}
\]
which is in $\spA{k+1}$. Similarly, for the function case, we have
\[
\Xi_x^y |\vn_{x_1,\ldots,x_{k}}^{y_1,\ldots,y_{k}}\> =
\begin{cases}
|\vn_{x_1,\ldots,x_{k}}^{y_1,\ldots,y_{k}}\> & \text{if }\exists i\colon (x_i,y_i)=(x,y),\\
\frac{1}{\sqrt{N}}|\vn_{x_1,\ldots,x_{k},x}^{y_1,\ldots,y_{k},y}\> & \text{if }\forall i\colon x_i\ne x, \\
0 & \text{if }\exists i\colon (x_i= x\And y_i\ne y) .
\end{cases}
\]

\paragraph{Overall success probability (proof framework for Point 2).}

First, note that
\begin{multline*}
\big\|P|\psi_k\>\big\|
\le \big\|P\widetilde\Pi^\lo_{k}|\psi_k\>\big\| + \big\|P\widetilde\Pi^\hi_{k}|\psi_k\>\big\|
\\
\le \big\|P\widetilde\Pi^\lo_{k}\big\|\cdot\big\||\psi_k\>\big\| + \big\|P\big\|\cdot\big\|\widetilde\Pi^\hi_{k}|\psi_k\>\big\|
= \big\|P\widetilde\Pi^\lo_{k}\big\| + \alpha_k.
\end{multline*}
Recall Eqn.~(\ref{eq:PDecomp}) that expresses $P$ via projectors $\Xi_x^\zero$.
By Property~5 of Fact~\ref{fac:norms} we have
\[
\big\|P\widetilde\Pi^\lo_{k}\big\|
=
\Big\|\sum_{x\in\Dom}\Xi_x^\zero\Pi^\lo_{k}\otimes I_{\regW''}\otimes|x\>\<x|_\regR\Big\|
=
\big\|\Xi_x^\zero\Pi^\lo_{k}\big\|.
\]
For functions, we evaluate this norm in Section~\ref{ssec:funcProofs}. For permutations, we evaluate the norm in Sections~\ref{sec:RepPrelim}--\ref{sec:MainTmProof} using the representation theory.

\paragraph{Amplitude transfer from ``low'' to ``high'' subspaces (proof framework for Point 3).}

Point 1 of the theorem states
\[
\big(\widetilde\Pi^\lo_{k-1}+\widetilde\Pi^\hi_{k-1}\big)|\psi_{k-1}\>=|\psi_{k-1}\>.
\]
By applying the triangle inequality, 
\begin{equation}
\label{eqn:alphaGap}
\begin{split}
\alpha_k 
 = &\,\big\|\widetilde\Pi^\hi_{k}\widetilde{O}|\psi_{k-1}\>\big\|
\\ \le &\,\big\|\widetilde\Pi^\hi_{k}\widetilde{O}\widetilde\Pi^\hi_{k-1}|\psi_{k-1}\>\big\|
 + \big\|\widetilde\Pi^\hi_{k}\widetilde{O}\widetilde\Pi^\lo_{k-1}|\psi_{k-1}\>\big\|
\\ \le &\,\big\|\widetilde\Pi^\hi_{k}\widetilde{O}\big\|\cdot\big\|\widetilde\Pi^\hi_{k-1}|\psi_{k-1}\>\big\|
 + \big\|\widetilde\Pi^\hi_{k}\widetilde{O}\widetilde\Pi^\lo_{k-1}\big\|\cdot\big\||\psi_{k-1}\>\big\|
\\ = &\,1\cdot\alpha_{k-1}
 + \big\|(\Pi^\hi_{k}\otimes I_{\regI\regO})O(\Pi^\lo_{k-1}\otimes I_{\regI\regO})\big\|\cdot1.
\end{split}
\end{equation}
For the last norm above, we use the expression (\ref{eq:ODecomp}) for $O$ to obtain\begin{multline}
\label{eq:NormSimpl}
\big\|(\Pi^\hi_{k}\otimes I_{\regI\regO})O(\Pi^\lo_{k-1}\otimes I_{\regI\regO})\big\|
\\ =
\Big\|\big(\Pi^\hi_{k}\otimes I_{\regI\regO}\big)
\sum_{x,y,y'}
\big(\Xi_x^y\otimes|x\>\<x|\otimes|y'\oplus y\>\<y|\big)
\big(\Pi^\lo_{k-1}\otimes I_{\regI\regO}\big)\Big\|
\\=
\Big\|\sum_{y,y'}
\big(\Pi^\hi_{k}\Xi_x^y\Pi^\lo_{k-1}\otimes|y'\oplus y\>\<y'|\big)\Big\|.
\end{multline}

For the function case, in Section~\ref{ssec:funcProofs}, we will be able to relatively easily show that (\ref{eq:NormSimpl})  equals $\sqrt{2/N}$. While the following framework still holds for the function case as well, its main purpose is its application for the permutation case.

To upper bound the rightmost norm of  (\ref{eq:NormSimpl}), let us rewrite the operator under the norm as a sum of two slightly more complex operators, yet whose norms are easier to bound.
This decomposition is inspired by decomposing a single call to the oracle as a sequence of a computing oracle call and an uncomputing oracle call (for example, as in~\cite{amrr:index}); see Figure~\ref{fig:CompUncomp}. Such a decomposition introduces another copy of the oracle output register $\regO$, and we implicitly do the same in the equality (\ref{eqn:CompUncompAll}) below.

\begin{figure}[ht!]
\begin{tikzcd}
\lstick{$\ket{x}_{\regI}$\hspace{6.2pt}} & \gate[wires=2]{O_f} &  \qw\rstick{$\ket{x}_{\regI}$} \\
\lstick{$\ket{y'}_{\regO}$\hspace{0.4pt}} & & \qw\rstick{$\ket{y'\oplus f(x)}_{\regO}$} \\
\lstick{$\ket{\mathbf{0}}_{\regO'}$} & \qw & \qw\rstick{$\ket{\mathbf{0}}_{\regO'}$} \\
\end{tikzcd}
\quad$=$\quad
\begin{tikzcd}
\lstick{$\ket{x}_{\regI}$\hspace{6.2pt}} & \qw & \gate[wires=2]{O_f} & \qw & \gate[wires=2]{O_f} & \qw & \qw\rstick{$\ket{x}_{\regI}$}
\\
\lstick{$\ket{y'}_{\regO}$\hspace{0.4pt}}  & \gate[swap]{} &  & \ctrl{1} & & \gate[swap]{}&\qw\rstick{$\ket{y'\oplus f(x)}_{\regO}$}
\\
\lstick{$\ket{\mathbf{0}}_{\regO'}$} & & \qw & \qw\oplus & \qw & &\qw\rstick{$\ket{\mathbf{0}}_{\regO'}$}
\end{tikzcd}
\caption{\small A single oracle call can be decomposed as a sequence of two oracle calls, so that, in the output register $\regO$, the first call always starts with $\ket{\mathbf{0}}$ and results in $\ket{f(x)}$ and the second call always starts with $\ket{f(x)}$ and results in $\ket{\mathbf{0}}$.
They are called computing and uncomputing oracle calls, respectively.
}
\label{fig:CompUncomp}
\end{figure}
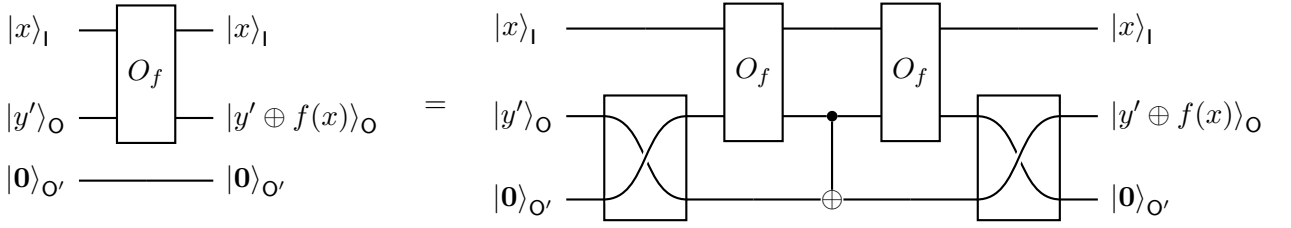

We have $\Xi_x^y=\Xi_x^y\Xi_x^y$ and, by Point~1, we have 
\[
\Xi_x^y\Pi^\lo_{k-1} =
(\Pi^\lo_{k}+\Pi^\hi_{k})\Xi_x^y\Pi^\lo_{k-1}.
\]
Hence, we can write
\[
\Pi^\hi_{k}\Xi_x^y\Pi^\lo_{k-1} = \Pi^\hi_{k}\Xi_x^y(\Pi^\hi_{k}+\Pi^\lo_{k})\Xi_x^y\Pi^\lo_{k-1}.
\]
Using this, we decompose the operator under the rightmost norm in (\ref{eq:NormSimpl}) as
\begin{subequations}
\label{eqn:CompUncompAll}
\begin{align}
&\sum_{y,y'} \big(\Pi^\hi_{k}\Xi_x^y\Pi^\lo_{k-1}\otimes|y'\oplus y\>\<y'|\big)
\\\label{eqn:CompUncompB}
& =
\sum_{y,y'} \big(\Pi^\hi_{k}\Xi_x^y \otimes|y'\oplus y\>\<y'|\otimes\<y|\big)
\sum_{y} \big(\Pi^\hi_{k}\Xi_x^y\Pi^\lo_{k-1}\otimes I_N\otimes|y\>\big)
\\\label{eqn:CompUncompC}
& \quad +
\sum_{y} \big(\Pi^\hi_{k}\Xi_x^y\Pi^\lo_{k}\otimes I_N\otimes\<y|\big)
\sum_{y,y'} \big(\Xi_x^y\Pi^\lo_{k-1} \otimes|y'\oplus y\>\<y'|\otimes|y\>\big),
\end{align}
\end{subequations}
where $I_N$ is the $N$-dimensional identity matrix.
Let us separately bound the norms of the four sums in (\ref{eqn:CompUncompB}) and (\ref{eqn:CompUncompC}). For the former sum in (\ref{eqn:CompUncompB}), we have
\[
\Big\|\sum_{y,y'} \big(\Pi^\hi_{k}\Xi_x^y \otimes|y'\oplus y\>\<y'|\otimes\<y|\big)\Big\|
\le
\big\|\Pi^\hi_{k}\otimes I_N\big\|
\Big\|\sum_{y,y'} \big(\Xi_x^y \otimes|y'\oplus y\>\<y'|\otimes\<y|\big)\Big\|=1,
\]
where the last equality holds because, for distinct $y$, $\Xi^y_x$ project on orthogonal spaces and because $\sum_{y'}|y'\oplus y\>\<y'|$ is a unitary. Similarly, the norm of the latter sum in (\ref{eqn:CompUncompC}) is at most $1$. Now consider the remaining two sums. We bound the norm of the latter sum in (\ref{eqn:CompUncompB}) as
\begin{align*}
& \Big\|
\sum_{y} \big(\Pi^\hi_{k}\Xi_x^y\Pi^\lo_{k-1}\otimes I_N\otimes|y\>\big)
\Big\|
=
\Big\|
\sum_{y} \big(\Pi^\lo_{k-1}\Xi_x^y\Pi^\hi_{k}\otimes \<y|\big)
\sum_{y} \big(\Pi^\hi_{k}\Xi_x^y\Pi^\lo_{k-1}\otimes |y\>\big)
\Big\|^{1/2}
\\ & \qquad
=
\Big\|
\sum_{y} \big(\Pi^\lo_{k-1} \Xi_x^y \Pi^\hi_{k} \Xi_x^y \Pi^\lo_{k-1}\big)
\Big\|^{1/2}
\le 
\Big(
\sum_{y} \big\|\Pi^\lo_{k-1} \Xi_x^y \Pi^\hi_{k} \Xi_x^y \Pi^\lo_{k-1}\big\|
\Big)^{1/2}
\\ & \qquad
=
\Big(
\sum_{y} \big\|\Pi^\lo_{k-1} \Xi_x^y \Pi^\hi_{k} \big\|^2
\Big)^{1/2}
\le
\Big(
\sum_{y} \big\|\Pi^\lo_k \Xi_x^y \Pi^\hi_{k} \big\|^2
\Big)^{1/2},
\end{align*}
where the last inequality is from Property~3 of Fact~\ref{fac:norms}.
We obtain the same upper bound on the former sum (\ref{eqn:CompUncompC}) analogously. 
By the triangle inequality, we have thus obtained
\begin{equation}
\label{eqn:GBTriangle}
\begin{split}
\Big\|\big(\Pi^\hi_{k}\otimes I_{\regI\regO}\big)
O
\big(\Pi^\lo_{k-1}\otimes I_{\regI\regO}\big)\Big\|
\, & \le
2\Big(
\sum_{y} \big\|\Pi^\lo_{k} \Xi_x^y \Pi^\hi_{k} \big\|^2
\Big)^{1/2}.
\end{split}
\end{equation}

In Section~\ref{sec:ZhandryApproach}, we show that the right-hand side of (\ref{eqn:GBTriangle}) evaluates to approximately $2\sqrt{2/N}$
 in the function case.
However, in the same section, we directly evaluate already (\ref{eq:NormSimpl}), giving the twice-as-strong bound of $\sqrt{2/(N-1)}$.
For the permutation case, evaluating the right-hand side of (\ref{eqn:GBTriangle}) seems to be much harder, and we bound it in Sections~\ref{sec:RepPrelim}--\ref{sec:MainTmProof} using the representation theory.

\section{Connection to Zhandry's Compressed Oracle Framework for Random Functions}
\label{sec:ZhandryApproach}

Recall that, for the Unstructured Search problem on random functions, $\Func$ is the set all functions $f\colon\Dom\rightarrow\Ran$. Here we introduce an alternative definitions for the high and low success subspaces $\cH^\hi$ and $\cH^\lo$ for the function case, and illuminate the connection to Zhandry's framework. We also conclude the proof of Theorem~\ref{thm:mainThm} for the function case.

\subsection{Tensor-Product Structure of the Oracle Space}

\paragraph{On Zhandry's compressed oracle framework.}

One of the problems studied in our work is Unstructured Search on random functions. There are many other problems on random functions one may consider, for example, the Collision Finding problem, which asks to find two distinct $x,x'\in\Dom$ such that $f(x)=f(x')$. For such problems, Zhandry provided a framework for upper bounding success probability of a query algorithm~\cite{zhandry:record}. Here we elaborate on Zhardry's approach for our problem of interest, the Unstructured Search problem.

In essence, Zhandry considered an oracle that has an $(N+1)$-dimensional quantum register $\regf^+_x$ for each input $x\in\Dom$, thus the total quantum memory of the oracle being $(N+1)^M$-dimensional. The space corresponding to each $\regf^+_x$ has the standard orthonormal basis $\{|y\>\colon y\in\Ran\}\cup\{\ketPerp\}$ where $\ketPerp$ is the special state indicating that the algorithm has not learnt the value of $f(x)$ (or, if it has learnt it, then it has subsequently unlearnt it). The oracle starts in state $\bigotimes_{x\in\Dom}\ketPerp_{\regf^+_x}$.

Assuming we are at a point where value $f(x)$ has not been learnt so far and the oracle is asked for this value, the oracle can uniformly at random decide on this value. Note that the state $\sum_{y\in\Ran}|y\>/\sqrt{N}$ naturally corresponds to such a random choice. And, indeed, if the oracle is about to query $f(x)$, Zhandry's method proceeds by ``sandwiching'' the oracle call between two calls of the operation that on register $\regf^+_x$ swaps states $\ketPerp$ and $\sum_{y\in\Ran}|y\>/\sqrt{N}$ and leaves the remaining $N-1$ dimensions of register $\regf^+_x$ as well as the states of all the other registers of the oracle fixed. This ``uncompressing/compressing'' procedure ensures that, when the oracle call copies the value $y=f(x)$ from register $\regf^+_x$ to the memory of the algorithm (as in Fig.~\ref{fig:oracle}), the state of $\regf^+_x$ is orthogonal to $\ketPerp$. The procedure is also very important for cryptographic applications where one requires a time efficient simulation of the oracle.

To upper bound the success probability for Unstructured Search, Zhandry observed the following. Let $p_k$ be the probability that, if one measured all the registers $\regf^+_x$ of the oracle after $k$ oracle calls in the standard basis, at least one of them would measure to state $|\zero\>$. Then we have%
\footnote{At the first glance, it may seem that the proof of \cite[Theorem~1]{zhandry:record} suggests that the increase in the square root of the approximate success probability, i.e., $\sqrt{p_{k+1}}-\sqrt{p_k}$, is at most $1/\sqrt{N}$. However, that proof has a typo: for two mutually orthogonal projectors $Q$ and $R$ and a unit vector $|\psi\>$, $\|Q|\psi\>\|+\|R|\psi\>\|$ is at most $\sqrt{2}$, not at most $1$.}
\begin{subequations}
\label{eqn:boundZhandry}
\begin{align}
\label{eqn:boundZhandryA}
& \sqrt{p_0}=0; \\
\label{eqn:boundZhandryB}
& \text{the success probability of the algorithm is at most }(\sqrt{p_\kappa}+1/\sqrt{N})^2; \\
\label{eqn:boundZhandryC}
& \sqrt{p_{k+1}}\le \sqrt{p_k}+\sqrt{2/N}.
\end{align}
\end{subequations}
We refer to $p_k$ as the \emph{approximate success probability}.

Notice that, if we associate $\sqrt{p_k}$ in (\ref{eqn:boundZhandry}) with $\alpha_k$ in Theorem~\ref{thm:mainThm}, (\ref{eqn:boundZhandryB}) and (\ref{eqn:boundZhandryC}) are almost equivalent to, respectively, Points 2 and 3 of Theorem~\ref{thm:mainThm}. This correspondence is not coincidental.

\paragraph{Our framework.}

 As discussed in Section~\ref{sec:FuncReg}, in our approach, the memory of the oracle is $M^N$-dimensional space $\cF$ with orthonormal basis $\{ |f\>_\regF\colon f\in\Func\}$. 
We can follow a very similar path to that of Zhandry by decomposing $\cF$ as a tensor product of $M$ subspaces, except without introducing the extra dimension to these subspaces to accommodate state $\ketPerp$.  That is, we decompose the quantum memory $\regF$ into $M$ (sub)registers $\regf_x$, where $x\in\Dom$ and each $\regf_x$ corresponds to an $N$-dimensional space with the standard basis $\{|y\>\colon y\in\Ran\}$.
  Thereby, we can represent state $|f\>_\regF$ by each register $\regf_x$ holding one value of function $f$, namely, 
\[
|f\>_\regF = \bigotimes_{x\in\Dom}|f(x)\>_{\regf_x}.
\]
Not having the state $\ketPerp$, we also do not use Zhandry's ``uncompressing/compressing'' procedure.

We interpret register $\regf_x$ being in state 
\[
|\Fzero\>:=\frac{1}{\sqrt{N}}\sum_{y\in\Ran}|y\>
\]
 to mean that the algorithm has not learnt $f(x)$. While in Zhandry's approach, states $\ketPerp_{\regf^+_x}$ and $|\zero\>_{\regf^+_x}$ were orthogonal, in our case $|\Fzero\>_{\regf_x}$ and $|\zero\>_{\regf_x}$ have inner product $1/\sqrt{N}$. Therefore, to prove a similar statement as above, for defining approximate success probability $p_k$, we would not want to define $p_k$ by measuring $\regf_x$ on $|\zero\>$ and the space orthogonal to it. Instead, consider the state
\[
|\vecr\> := \sqrt{\frac{N-1}{N}}|\zero\> - \frac{1}{\sqrt{N(N-1)}}\sum_{y\ne\zero}|y\>,
\]
which is the closest state to $|\zero\>$ that is orthogonal to $|\Fzero\>$ and consider a two outcome measurement that measures $\regf_x$  on the one-dimensional space $\cR_0:=\spn\{|\vecr\>\}$ and the $(N-1)$-dimensional space orthogonal to $|\vecr\>$; let's call the latter space $\cR_1$. Let us show that, if we defined $p_k$ to be the probability that at least for one value of $x$ the register $\regf_x$ measures to $\cR_0$, we would obtain the statement essentially equivalent to the function case of Theorem~\ref{thm:mainThm}, and, as a result,  (\ref{eqn:boundZhandry}) with $\sqrt{2/(N-1)}$ instead of $\sqrt{2/N}$ in (\ref{eqn:boundZhandryC}) would be satisfied. Such equivalence is established by the following claim. 

\begin{clm}
\label{clm:HLalt}
In the function case, we have 
\[
\cH^\hi = \bigoplus_{\substack{b\colon\Dom\rightarrow\{0,1\}\\\exists x\colon b(x)=0}}\bigotimes_{x\in\Dom}(\cR_{b(x)})_{\regf_x}
\qquad\text{and}\qquad
\cH^\lo = \bigotimes_{x\in\Dom}(\cR_1)_{\regf_x}.
\]
\end{clm}

This claim is a corollary of Lemma~\ref{lem:HLaltGen}. Before we state and prove Lemma~\ref{lem:HLaltGen}, let us introduce some notation required for the lemma and for the conclusion of the proof of the function case of Theorem~\ref{thm:mainThm}.

\subsection{Alternative Definitions for High and Low Subspaces}

\paragraph{States and operators for a single function-value register $\regf_x$.}

Consider any $x$, its corresponding register $\regf_x$, which in turn corresponds to $N$-dimensional space $\cH$ with orthonormal basis $\{|y\>\colon y\in\Ran\}$.
For each such register, we define the following states in $\cH$ and operators on $\cH$.

Consider the standard basis state $|\zero\>$. Let $E_0:=|\zero\>\<\zero|$ be the projector on $|\zero\>$, and let $E_1:=I_N-E_0$.
Recall the \emph{Fourier zero-state} defined as
$|\Fzero\> = \sum_{y\in\Ran}|y\>/\sqrt{N}$.
Let $\FE_0:=|\Fzero\>\<\Fzero|$ be the projector on $|\Fzero\>$,
and let $\FE_1:=I_N-\FE_0$.
Let $\Erest$ be the projector on the $(N-2)$-dimensional $
(+1)$-eigenspace of $E_1\FE_1 E_1$, that is, the space orthogonal to both $|\zero\>$ and $|\Fzero\>$.

Recall the vector $|\vecr\>$ and spaces $\cR_0$ and $\cR_1$, and note that we have
\[
|\vecr\> = \FE_1|\zero\>\big/\|\FE_1|\zero\>\| = \FE_1|\zero\>\big/\sqrt{1-1/N}.
\]
Let $R_0:=|\vecr\>\<\vecr|$ and $R_1:=I_N-R_0$ be the projectors on $\cR_0$ and $\cR_1$, respectively. Note that $\FE_1 = |\vecr\>\<\vecr| + \Erest$.

\paragraph{States and operators for the whole function register $\regF$.}

For a subset $S\subseteq\Dom$, let
\[
\FE_S :=
\bigotimes_{x\in S}(\FE_1)_{\regf_x}
\otimes
\bigotimes_{x\notin S}(\FE_0)_{\regf_x},
\]
where $x\notin S$ is short for $x\in\Dom\setminus S$.
$\FE_S$ acts on the whole space $\cF$ of the oracle, and note that 
$I_\regF = \sum_{S\subseteq\Dom}\FE_S$.
Define
\[
\FE^{(\Dom)}_{k} := \sum_{\substack{S\subseteq \Dom\\|S|=k}}\FE_S.
\]
For the following lemma, recall that $\overline\cH_k^\lo=\cA_k\cap(\cB_k+A_{k-1})^\perp$, $\overline\cH_k^\hi=\cB_k\cap\cA_{k-1}^\perp$, and $\overline\cH_k^\lo\oplus\overline\cH_k^\hi=\spA{k}\cap \spA{k-1}^\perp$.

\begin{lem}
\label{lem:HLaltGen}
We have
\begin{alignat*}{2}
& \overline\Pi_k^\lo
& = & \sum_{\substack{S\subseteq\Dom\\|S|=k}} 
\bigg(
 \bigotimes_{x\in S}(\Erest)_{\regf_x}
\otimes
\bigotimes_{x\notin S}(\FE_0)_{\regf_x}
\bigg),
\\
& \overline\Pi_k^\hi &\; = &
\sum_{\substack{S\subseteq\Dom\\|S|=k}} 
\bigg(
 \Big(\bigotimes_{x\in S}(\FE_1)_{\regf_x}-\bigotimes_{x\in S}(\Erest)_{\regf_x}\Big)
\otimes
\bigotimes_{x\notin S}(\FE_0)_{\regf_x}
\bigg),
\end{alignat*}
and therefore $\overline\Pi_k^\lo+\overline\Pi_k^\hi = \FE^{(\Dom)}_{k}$.
\end{lem}

\begin{proof}
Let us start by finding the orthogonal projector on $\spA{k}=\cH_k^\lo\oplus\cH_k^\hi$. 
Fix $S\subseteq \Dom$. There are $N^{|S|}$ assignments $\alpha$ with domain $S$, and for all of them vectors $|v_\alpha\>$ are mutually orthogonal.
Therefore,
\[
\sum_{\alpha\colon\mathrm{dom}\,\alpha=S}
|v_\alpha\>\<v_\alpha|
=
\bigotimes_{x\in S}I_{\regf_x}
\otimes
\bigotimes_{x\notin S}(\FE_0)_{\regf_x}
=
\sum_{T\subseteq S} \FE_T,
\]
and the orthogonal projector on $\spA{k}$ is the support of 
\[
\sum_{\substack{S\subseteq\Dom\\|S|=k}} \sum_{T\subseteq S} \FE_T
= \sum_{k'=0}^k \binom{M-k'}{k-k'}\sum_{\substack{T\subseteq\Dom\\|T|=k'}} \FE_T
= \sum_{k'=0}^k \binom{M-k'}{k-k'} \FE_{k'}^{(\Dom)}.
\]
Namely, the orthogonal projector on $\spA{k}$ is $\sum_{k'=0}^k \FE_{k'}^{(\Dom)}$ and, in turn, the orthogonal projector on $\spA{k}\cap\spA{k-1}^\perp$ is $\FE_{k}^{(\Dom)}$, as claimed.

Now let us find the orthogonal projector on $\cH_k^\hi=\spB{k}\cap\spA{k-1}^\perp$.
Fix $S\subseteq \Dom$ of size $k$, and define
\[
\Psi_S:=
\sum_{\substack{\alpha\\\mathrm{dom}\,\alpha=S\\\zero\in\mathrm{im}\,\alpha}}
|v_\alpha\>\<v_\alpha|.
\]
 There are $N^{|S|}-(N-1)^{|S|}$ assignments $\alpha$ with domain $S$ such that $\zero\in\mathrm{im}\,\alpha$, and for all of them vectors $|v_\alpha\>$ are mutually orthogonal.
Therefore,
\[
\Psi_S
=
\Big(\bigotimes_{x\in S}I_{\regf_x}-\bigotimes_{x\in S}(E_1)_{\regf_x}\Big)
\otimes
\bigotimes_{x\notin S}(\FE_0)_{\regf_x}.
\]
Note that $\Psi_S$ is orthogonal to every $\FE_{S'}$ with $|S'|=k$ except when $S'=S$.

Note that $\cB_k$ equals the support of $\sum_{\substack{S\subseteq\Dom\\|S|=k}}\Psi_S$.
Since $\cB_k\subseteq\cA_k$ and the orthogonal projector on $\cA_k\cap\cA_{k-1}^\perp$ is $\FE^{(\Dom)}_{k} = \sum_{\substack{S\subseteq \Dom\\|S|=k}}\FE_S$,
the space $\cB_k\cap\cA_{k-1}^\perp$ equals the support of 
\begin{multline*}
\FE^{(\Dom)}_{k}
\bigg(\sum_{\substack{S\subseteq\Dom\\|S|=k}}\Psi_S\bigg)
\FE^{(\Dom)}_{k}
=
\sum_{\substack{S\subseteq\Dom\\|S|=k}}
\FE_S \Psi_S \FE_S
\\ =
\sum_{\substack{S\subseteq\Dom\\|S|=k}}
\bigg(
\Big(\bigotimes_{x\in S}(\FE_1)_{\regf_x}-\bigotimes_{x\in S}(\FE_1E_1\FE_1)_{\regf_x}\Big)
\otimes
\bigotimes_{x\notin S}(\FE_0)_{\regf_x}
\bigg),
\end{multline*}
where we have used the observation that $\FE_S\Psi_{S'}=0$ for distinct $S$ and $S'$ of size $k$.

Note that $\FE_1 E_1\FE_1 = \Erest + \frac{1}{N}|\vecr\>\<\vecr|$.
Because $\FE_1 = \Erest + |\vecr\>\<\vecr|$, both 
$\bigotimes_{x\in S}(\FE_1)_{\regf_x}$ and $\bigotimes_{x\in S}(\FE_1E_1\FE_1)_{\regf_x}$ have the same support, and the non-zero eigenvectors of their difference are exactly those eigenvectors of $\bigotimes_{x\in S}(\FE_1E_1\FE_1)_{\regf_x}$ whose eigenvalues are strictly between $0$ and $1$.
One can see that the projector on such eigenvectors is $\bigotimes_{x\in S}(\FE_1)_{\regf_x}-\bigotimes_{x\in S}(\Erest)_{\regf_x}$.
Thus, in turn, the projector on $\overline\cH_k^\hi=\cB_k\cap\cA_{k-1}^\perp$ is 
\[
\sum_{\substack{S\subseteq\Dom\\|S|=k}} 
\bigg(
 \Big(\bigotimes_{x\in S}(\FE_1)_{\regf_x}-\bigotimes_{x\in S}(\Erest)_{\regf_x}\Big)
\otimes
\bigotimes_{x\notin S}(\FE_0)_{\regf_x}
\bigg).
\]
As a result, the projector on $\overline\cH_k^\lo=\cA_k\cap(\cB_k+A_{k-1})^\perp$ is
\[
 \FE_{k}^{(\Dom)}-\overline\Pi_k^\hi 
= \sum_{\substack{S\subseteq\Dom\\|S|=k}} 
\bigg(
 \bigotimes_{x\in S}(\Erest)_{\regf_x}
\otimes
\bigotimes_{x\notin S}(\FE_0)_{\regf_x}
\bigg).
\qedhere
\]
\end{proof}

To see that Lemma~\ref{lem:HLaltGen} implies Claim~\ref{clm:HLalt}, note that $\Erest+\FE_0=R_1$. As a result, we have
\[
\Pi^\lo 
= \sum_{k=0}^M \overline\Pi_k^\lo 
= \sum_{S\subseteq\Dom} 
\bigg(
 \bigotimes_{x\in S}(\Erest)_{\regf_x}
\otimes
\bigotimes_{x\notin S}(\FE_0)_{\regf_x}
\bigg)
= \bigotimes_{x\in\Dom}(R_1)_{\regf_x}.
\]

\subsection{Conclusion of the Proof of Theorem~\ref{thm:mainThm} for the Function Case}
\label{ssec:funcProofs}

In the function case, Theorem~\ref{thm:mainThm} follows from the following two claims.

\begin{clm}
\label{clm:funcLowBound}
In the function case, we have $\big\|\Xi_x^\zero\Pi^\lo_k\big\|
\le 
1/\sqrt{N}$.
\end{clm}

\begin{clm}
\label{clm:funcTransfBound}
In the function case, we have
\[
\Big\|\sum_{y,y'}
\big(\Pi^\hi_{k}\Xi_x^y\Pi^\lo_{k-1}\otimes|y'\oplus y\>\<y'|\big)\Big\| < \sqrt{2/(N-1)}.
\]
\end{clm}

In the proofs of both claims we use Property~3 of Fact~\ref{fac:norms}, which lets us replace $\Pi^\lo_k$ by $\Pi^\lo$ and $\Pi^\hi_k$ by $\Pi^\hi$. In addition, note that
\[
\Xi_x^y= |y\>\<y|_{\regf_x}\otimes\bigotimes_{x'\in\Dom\setminus\{x\}} I_{\regf_{x'}}.
\]

\begin{proof}[Proof of Claim~\ref{clm:funcLowBound}]
We have
\[
\Xi_x^\zero\Pi^\lo
= (|\zero\>\<\zero|R_1)_{\regf_x}\otimes\bigotimes_{x'\in\Dom\setminus\{x\}} (R_1)_{\regf_{x'}},
\]
therefore
\[
\|\Xi_x^\zero\Pi^\lo\|^2
=\|\<\zero|R_1\|^2
= 1-|\<\zero|\vecr\>|^2 = 1/N. \qedhere
\]
\end{proof}

\begin{proof}[Proof of Claim~\ref{clm:funcTransfBound}]
Recall that we interpret elements of range $\Ran$ as $\log_2N$-bit strings. Consider the Fourier basis of the space corresponding to the register $\regf_x$: for $z\in\Ran$, let
\[
|\widehat{z}\> := \frac{1}{\sqrt{N}}\sum_{y'}(-1)^{z\cdot y'}|y'\>
\]
(which is consistent with our earlier notation of $|\Fzero\>$).
In the Fourier basis, we have
\[
\sum_{y'}|y'\oplus y\>\<y'| = \sum_{z}(-1)^{y\cdot z}|\widehat{z}\>\<\widehat{z}|,
\]
and therefore
\[
\Big\|\sum_{y,y'}\big(\Pi^\hi\Xi_x^y\Pi^\lo\otimes|y'\oplus y\>\<y'|\big)\Big\|
=
\max_{z}
\Big\|\Pi^\hi\Big(\sum\nolimits_y(-1)^{y\cdot z}\Xi_x^y\Big)\Pi^\lo\Big\|,
\]
where the summation and the maximum is over $z\in\Ran$.%
\footnote{This equality is also true for the permutation case, however, it is not clear how to evaluate this norm in that case.}
If $z=0$, the operator under the norm is $\Pi^\hi I_\regF\Pi^\lo=0$, so let us assume $z\ne 0$.
We have
\[
\Big\|\Pi^\hi\Big(\sum\nolimits_y(-1)^{y\cdot z}\Xi_x^y\Big)\Pi^\lo\Big\|
= \Big\||\vecr\>\<\vecr|\Big(\sum\nolimits_y(-1)^{y\cdot z}|y\>\<y|\Big) R_1\Big\|.
\]
Note that $\sum\nolimits_y(-1)^{y\cdot z}|y\>\<y|$ is a unitary and $R_1=I-|\vecr\>\<\vecr|$, so
\begin{multline*}
\Big\||\vecr\>\<\vecr|\Big(\sum\nolimits_y(-1)^{y\cdot z}|y\>\<y|\Big) R_1\Big\|^2
= 1-\Big|\<\vecr|\Big(\sum\nolimits_y(-1)^{y\cdot z}|y\>\<y|\Big) |\vecr\>\Big|^2
\\ = 1-\bigg(\frac{N-1}{N}+\frac{N/2-1}{N(N-1)}-\frac{N/2}{N(N-1)}\bigg)^2
 = \frac{2N-3}{(N-1)^2}
< \frac{2}{N-1},
\end{multline*}
because $|\<r_0|\zero\>|^2=(N-1)/N$ and $|\<r_0|y\>|^2=1/(N(N-1))$ for $y\ne\zero$.
\end{proof}

We note that, in the proof of Point 3 of Theorem~\ref{thm:mainThm}, if we bounded Eqn.~(\ref{eqn:GBTriangle}) instead of Eqn.~(\ref{eq:NormSimpl}), we would obtain the following weaker bound on $\alpha_k-\alpha_{k-1}$.
For all $y\in\Ran$, we have 
\[
\Pi^\lo\Xi_x^y\Pi^\hi
=
(R_1|y\>\<y|R_0)_{\regf_x}\otimes \bigotimes_{x'\in\Ran\setminus\{x\}} (R_1)_{\regf_{x'}}.
\]
As a result,
\[
\|\Pi^\lo\Xi_x^y\Pi^\hi\|
= \|R_1|y\>\|\cdot\|R_0|y\>\| 
= |\<\vecr|y\>|\sqrt{1-|\<\vecr|y\>|^2}
= \begin{cases}
\frac{\sqrt{N-1}}{N} & \text{if }y=\zero,\\
\frac{\sqrt{N^2-N-1}}{N(N-1)} & \text{if }y\neq\zero.\\
\end{cases}
\]
Finally,
\[
2\Big(
\sum_{y} \big\|\Pi^\lo \Xi_x^y \Pi^\hi \big\|^2
\Big)^{1/2}
=
2\Big(
\frac{N-1}{N^2}
+(N-1)\frac{N^2-N-1}{N^2(N-1)^2}
\Big)^{1/2}
 =
2\sqrt{\frac{2N-3}{N(N-1)}}.
\]

\section{Representation Theory}
\label{sec:RepPrelim}

In this section, let us present the basics of the representation theory of the symmetric group. 
For a detailed study of the representation theory of the symmetric group, refer to~\cite{james:symmetric,sagan:symmetric};
for the fundamentals of the representation theory of finite groups, refer to~\cite{serre:representation}.

\subsection{Representation Theory of Finite Groups}

Let $\cV$ be a complex Euclidean space and let $G$ be a finite group.
A \emph{(linear) representation} of $G$ on $\cV$ is a group homomorphism $V$ from $G$ to the general liner group $GL(\cV)$. We typically write $V_g$ instead of $V(g)$, call it a 
\emph{representation operator}, and, by definition, $V$ satisfies $V_{gg'}=V_gV_g'$ for all $g,g'\in G$.
We call $\dim\cV$ the dimension of the representation.
We may also refer to $\cV$ as a representation, when the operators $V_g$ are clear from the context, and say that $G$ \emph{acts} on $\cV$.

We call a representation $V$ on $\cV$ \emph{reducible} if there is a proper subspace $\cV'\subset\cV$ such that $\cV'\neq\{0\}$ and $\cV'$ is invariant under to $V_g$ for all $g\in G$.
We call $V$ irreducible (or, \emph{irrep}, for short) if it is not reducible.

We call two representations $V,V'$ \emph{isomorphic} and write $V\cong V'$ if there is an isomorphism $M$ such that $V_gM=MV'_g$ for all $g$.
If $V$ and $V'$ are irreducible and $M$ is an isomorphism between them, then Schur's lemma implies that $M=0$ whenever $V\not\cong V'$.
The number of mutually non-isomorphic irreps of $G$ equals the number of conjugacy classes of $G$, and squares of their dimensions sum up to $|G|$.

If two irreps $\rho$ and $\sigma$ are isomorphic, we may call $\rho$ an instance of $\sigma$.
Given a representation $V$ of $G$ on $\cV$ and an irrep $\rho$ of $G$, the $\rho$-isotypic subspace of $\cV$ is the subspace that corresponds to all instances of $\rho$ contained in $\cV$.
The following lemma is a consequence of Schur's lemma (see~\cite{amrr:index}). 

\begin{lem}\label{lem:isotypicOverlap}
Suppose we are given an isotypic subspace, three $\cH,\cH',\cH''$ isomorphic irreps of dimension $d$ belonging to that subspace, and orthogonal projectors $\Pi,\Pi',\Pi''$ on these irreps, respectively.
We have
$\tr[\Pi\Pi'\Pi\Pi'']=\tr[\Pi\Pi']\tr[\Pi\Pi'']/d$.
\end{lem}

Given two finite groups $G$ and $H$, any irrep of their direct product $G\times H$ is isomorphic to a representation $V_g\otimes V'_h$ where $V$ and $V'$ are irreps of $G$ and $H$, respectively.

If $V$ is a representation of $G$ and $H$ is a subgroup of $G$, $V$ restricted to the group $H$ is also a representation of $H$. We denote it $\Res^G_HV$, and it is not necessarily irreducible for $H$ even if it were irreducible for $G$.
Let $\mathrm{Inv}_H^GV'$ denote a representation of $G$ that is obtained by inducing a representation $V'$ of $H$, which is unique up to isomorphism.

A character of a representation $V$ is the complex-valued function
\[
\chi\colon G\rightarrow \mathbb{C}\colon g\mapsto  \tr[V_g].
\]
Given two characters $\chi$ and $\chi'$ of a group $G$, their inner product is
\[
\<\chi,\chi'\>:=\frac{1}{|G|}\sum_{g\in G}\overline\chi(g)\chi'(g).
\]
The inner product is non-negative integer and, if $\chi$ is a character of an  irrep $\rho$, then $\<\chi,\chi'\>$ equals the number of instances of $\rho$ contained by the representation corresponding to $\chi'$.
Given a representation $V$ on $\cV$ and a character $\chi$ of an irrep $\rho$,
\begin{equation}
\label{eqn:isotypicPi1}
\Pi_\rho:=\frac{\dim\rho}{|G|}\sum_{g\in G}\overline\chi_{g}V_g
\end{equation}
is the orthogonal projector on the $\rho$-isotypic subspace of $\cV$.

Consider the space $\mathbb{C}^G$ having an orthonormal basis $\{|g\>\colon g\in G\}$.
Maps $L\colon G\rightarrow GL(\mathbb{C}^G)$ and $R\colon G\rightarrow GL(\mathbb{C}^G)$ 
such that
\[
L_g\colon |h\> \mapsto |gh\>
\qquad\text{and}\qquad
R_g\colon |h\> \mapsto |hg^{-1}\>
\]
are know as, respectively, the \emph{left} and the \emph{right regular representation} of $G$.

\subsection{Representation Theory of the Symmetric Group} \label{sec:rep}

Let $\Sym_A$ be the symmetric group of a finite set $A$, that is, the group whose elements are all the
 permutations of elements of $A$ and whose group operation is the composition of permutations.

\paragraph{Partitions and Young diagrams.}

A partition of a non-negative integer $m$ is a tuple $\rho=(\rho_1,\ldots,\rho_r)$ of positive integers such that $\rho_1\ge \rho_2\ge \ldots \ge \rho_r$ and $\sum_{i=1}^r\rho_i=m$, and we write $\rho\vdash m$.
Given a partition $\rho=(\rho_1,\ldots,\rho_r)\vdash m$, its transpose partition is $\rho^\top:=(\rho^\top_1,\ldots,\rho^\top_c)\vdash m$, where 
\[
\rho^\top_j := \max\{i\colon \rho_i\ge j\}
\]
(note that $\rho^\top_1=r$ and $c=\rho_1$).

There is one-to-one correspondence between partitions of $m$ and $m$-box Young diagrams, and from now on we refer to both concepts as Young diagrams.
A \emph{Young diagram} $\rho=(\rho_1,\ldots,\rho_r)\vdash m$ consists $m$ square boxes (also known as cells) organized in $r$ left-aligned rows with $\rho_i$ boxes in row $i$; see Figure~\ref{fig:YoungEx}.
Note that the column $j$ of the Young diagram $\rho$ has $\rho^\top_j$ boxes.

\begin{figure}[ht!]
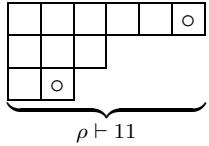

\centering
\[
\underbrace{\young(~~~~~\circ,~~~,~\circ)}_{\rho\;\vdash\; 11}
 \]
 \caption{\small An example of an $11$-box Young diagram $\rho=(6,3,2)$. The distance between the two marked vertices, $(3,2)$ and $(1,6)$, is $6$.}
\label{fig:YoungEx}
\end{figure}

Given a Young diagram $\rho$, let $|\rho|$ denote the number of boxes in it (as a result, $\rho\vdash|\rho|$). Let $\rho\in \lightPart{k}$ denote that $\rho\vdash\ell$ for some $\ell\le k$.

We identify a box by its row and column coordinates.
We say that a box $(i,j)$ is present in $\rho$ and write $(i,j)\in \rho$ if $\rho_i\ge j$ (equivalently, $\rho^\top_j\ge i$). 
The {\em distance} 
 between two boxes $(i,j)$ and $(i',j')$ is defined as
\(
|i'-i|+|j-j'|
\).

\paragraph{Specht modules and the hook length formula.}

Up to isomorphism, there is one-to-one correspondence between the irreps of $\Sym_A$ and $|A|$-box Young diagrams, and we denote the irrep corresponding to $\rho\vdash|A|$ by $\Spe{\rho}$, which is called the \emph{Specht module}.
The Specht module $\Spe{(|A|)}$ is the trivial representation that maps all permutations to $1$.
Let $\dm{\rho}:=\dim\Spe{\rho}$ denote the dimension of $\Spe{\rho}$, and we also sometimes denote it by $\dim\rho$. The dimension $\dm{\rho}$ is given by the \emph{hook length formula}.

Given a Young diagram $\rho\vdash m$ and a box $(i,j)\in\rho$, the \emph{hook length} of $(i,j)$ is
\[
h_\rho(i,j):= (\rho_i-j)+(\rho^\top_j-i)+1;
\]
see Figure~\ref{fig:HooksAll}.
That is, $h_\rho(i,j)$ equals the number of the boxes right from $(i,j)$ in the same row plus number of the boxes under $(i,j)$ in the same column plus $1$ (that is, the box itself).

\begin{figure}[ht!]
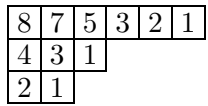

\centering
\[
\young(875321,431,21)
 \]
 \caption{\small Hook lengths of the $11$-box Young diagram $\rho=(6,3,2)$. The dimension of the Specht module $\Spe{\rho}$ is $\dm{\rho}=11!/(8\cdot7\cdot5\cdot3\cdot2\cdot1\cdot4\cdot3\cdot1\cdot2\cdot1)=990$.}
\label{fig:HooksAll}
\end{figure}

Let
\[
h(\rho):=\prod_{(i,j)\in\rho}h_\rho(i,j)
\]
be the product of the hook lengths of all boxes in $\rho$.
The dimension of the Specht module $\Spe{\rho}$ is
\[
\dm{\rho}=m!/h(\rho).
\]

\paragraph{Projectors on isotypic subspaces.}

Suppose $A$ is a finite set, $V$ is a unitary representation of $\Sym_A$ on $\cV$, and $\rho\vdash|A|$ is a Young diagram.
Given the character $\chi$ corresponding to $\rho$, expression (\ref{eqn:isotypicPi1}) provides one way to construct the orthogonal projector $\Pi_\rho$ on the $\Spe{\rho}$-isotypical subspace of $\cV$. 
Here we present another way to construct the same projector that does not require the knowledge of the character.
We also note that all irreducible characters of the symmetric group are real-valued.

A \emph{Young tableau} of shape $\rho$ is an assignment that assigns to every box of $\rho$ a unique element from $A$.%
\footnote{This is a slightly non-standard definition, as typically one considers $A=\{1,2,\ldots,|A|\}$. For our purposes, however, it is convenient to use the more general definition.}
 Clearly, there are $|A|!$ Young tableaux of shape $\rho$.
Given a subset $B\subseteq A$, define
\begin{alignat*}{2}
& E_B^+ := \sum_{\pi\in\Sym_B}V_\pi,
&\qquad&\widehat{E}_B^+ := E_B^+\big/|B|!,
\\& E_B^- := \sum_{\pi\in\Sym_B}\sgn\pi V_\pi,
&& \widehat{E}_B^- = E_B^-\big/|B|!,
\end{alignat*}
where $\sgn\pi=\pm 1$ is the sign of the permutation.
One can see that both $\widehat{E}_B^+$ and $\widehat{E}_B^-$ are orthogonal projectors.

Given a Young tableau $\tabl$, let $R_i(\tabl)$ and $C_j(\tabl)$ be the sets formed by entries in $i$-th row and $j$-th column of $\tabl$, respectively.
Let
\[
E^+_{\mathrm{row}}(\tabl)=\prod_i E^+_{R_i(\tabl)}
\qquad\text{and}\qquad
E^-_{\mathrm{col}}(\tabl)=\prod_j E^-_{C_j(\tabl)},
\]
where the products are over all the rows and all the columns, respectively.
Note that all $E^+_{R_i(\tabl)}$ mutually commute and so do all $E^-_{C_j(\tabl)}$. However, $E^+_{R_i(\tabl)}$ and $E^-_{C_j(\tabl)}$ in general do not commute.

Define the operator
\[
P_\tabl := \frac{\dm{\rho}}{|A|!}
E^-_{\mathrm{col}}(\tabl)E^+_{\mathrm{row}}(\tabl).
\]
This operator is a projector in the sense that $P_\tabl^2=P_\tabl$, but it is not necessarily an orthogonal projector, as it may differ from its transpose $P^\top_\tabl$, which is also a projector. However,
\begin{equation}
\label{eqn:PiRhoAlt}
\Pi_\rho
=\frac{\dm{\rho}}{|A|!}\sum_{\tabl}P_\tabl
=\frac{\dm{\rho}}{|A|!}\sum_{\tabl}P_\tabl^\top
\end{equation}
is an orthogonal projector (hence the equality to its transpose) that projects on the $\Spe{\rho}$-isotypic subspace,
where the sums are over all $|A|!$ Young tableaux $\tabl$ of shape $\rho$.%
\footnote{The same projector $\Pi_\rho$ would be obtained if one took the sum only over what is known as the \emph{standard} Young tableaux and dropped the normalization factor $\dm{\rho}/|A|!$.}

\paragraph{Branching rule.}

Let $\sigma\yunder\rho$ denote that a Young diagram $\sigma$ is obtained from $\rho$ by removing exactly one box. 
 The {\em branching rule} states that the restriction of an irrep $\Spe\rho$ of $\Sym_{A}$ to $\Sym_{A\setminus\{a\}}$, where $a\in A$, is
\[
\Res\nolimits_{\,\Sym_{A\setminus\{a\}}}^{\,\Sym_A} \!\Spe\rho \cong \bigoplus\nolimits_{\sigma\yunder\rho}\Spe\sigma.
\]
{\em{}Frobenius reciprocity} then tells us that the ``opposite'' happens when we induce an irrep of
 $\Sym_{A\setminus\{a\}}$ to $\Sym_A$.
 
We also use the following special case of the Littlewood--Richardson rule.
Suppose $A=A'\sqcup A''$ is a disjoint union of $A$. 
Given non-negative integers $k$ and $\ell > k$ and a Young diagram $\xi\vdash k$, let
$\Lambda_\ell^+(\xi)$ be the set of all $\ell$-box Young diagrams that can be obtained from $\xi$ by adding at most one box per each column (including previously empty columns).
Given $\xi\vdash|A|$, we have
 \begin{equation}
 \label{eqn:IndLR}
  \mathrm{Ind}_{\Sym_{A'}\times\Sym_{A''}}^{\Sym_A} \Spe{\xi}\otimes\Spe{(|A|)}
 = \bigoplus\nolimits_{\rho\in\Lambda^+_{|A|}(\xi)}\Spe{\rho}.
 \end{equation}
 We note that the branching rule is a special case of this when $A''=\{a\}$ because $\Sym_{A\setminus\{a\}}$ and $\Sym_{A\setminus\{a\}}\times\Sym_{\{a\}}$ are isomorphic.

\paragraph{Shorthand notation.}

Very often in this paper we treat boxes in a first row of a Young diagram differently than boxes in rows below it.
For that reason, let us introduce notation that can express such a treatment more concisely.

Given a Young diagram $\rho=(\rho_1,\rho_2,\ldots,\rho_r)$ and $m\ge\rho_1$, let 
$(m,\rho)$ be short for the Young diagram $(m,\rho_1,\rho_2,\ldots,\rho_r)$. 

Recall that, for the Inverting Permutations problem, $\Dom$ and $\Ran$ are, respectively, the domain and the range for bijections $f$, both having size $N$.
Given $\ell\in\{0,1,2\}$, a set $A=\Dom$ or $A=\Ran$, its subset $A\setminus\{a_1,\ldots,a_\ell\}$, and $\rho\vdash N-\ell$, let us write $\rho_{a_1\!\ldots a_\ell}$ if we want to stress that we think of $\Spe{\rho}$ as an irrep of $\Sym_{A\setminus\{a_1,\ldots,a_\ell\}}$. We omit the subscript if $\ell=0$ or when $\{a_1,\ldots,a_\ell\}$ is clear from the context. To lighten the notations, given $k\le N/2$ and $\parE\vdash k$, let $\parEb_{a_1\!\ldots a_\ell}=(N-\ell-k,\parE)_{a_1\!\ldots a_\ell}\vdash N-\ell$;
here we omit the subscript if and only if $\ell=0$. 
We may write $\bar\eta_*$ instead of $\bar\eta_{a_1}$ and $\bar\eta_{**}$ instead of $\bar\eta_{a_1a_2}$ where $a_1$ and $a_2$ are clear from the context or irrelevant, like when we write $\dim{\bar\eta_*}$.

We use
 $\parD$, $\parE$, and $\parF$ to denote Young diagrams having $o(N)$ boxes, $\parA$, $\parB$, and $\parC$ to denote Young diagrams having $N$, $N-1$, and $N-2$ boxes, respectively, and $\rho$ and $\sigma$ for general statements and other purposes.

\section{Decomposition of $\cF$ via Representation Theory}
\label{sec:ABviaRepTh}

Here we consider the permutation case only.
Recall that $|\Dom|=|\Ran|=N$, $\Func$ is the set of all $N!$ bijections $f\colon\Dom\rightarrow\Ran$, and $\cF=\mathbb{C}^{\Func}$ is the complex Euclidean space corresponding to $\Func$.
Permutations $\pi\in\Sym_\Dom$ of the domain and permutations $\tau\in\Sym_\Ran$ of the range act on bijections $f$ in the natural way:
\begin{subequations}
\label{eqn:ActOnF}
\begin{alignat}{3}
&\pi
 &\;\;:\;\;&f \;\;\mapsto\;\;
f_\pi:=f\circ\pi^{-1},\\
&\tau
&\;\;:\;\; &f \;\; \mapsto \;\;
f^\tau :=\tau\circ f.
\end{alignat}
\end{subequations}
Note that, since $f$ is a bijection, so are $f_\pi$ and $f^\tau$.
The actions of $\pi\in\Sym_\Dom$ and $\tau\in\Sym_\Ran$ commute: we have $(f_\pi)^\tau=(f^\tau)_\pi$, which we denote by $f^\tau_\pi$ for short.

Let $V_\pi^\tau$ be the unitary operator that acts on the space $\mathcal{F}$ and that maps every $|f\>$ to
$|f_\pi^\tau\>$.
Then $V\colon (\pi,\tau)\mapsto V_\pi^\tau$ is a unitary representation of $\Sym_\Dom\times\Sym_\Ran$, and it decomposes into irreps of $\Sym_\Dom\times\Sym_\Ran$ as the direct product
\begin{equation}
\label{eqn:LambdaLambdaDecom}
V \cong
\bigoplus_{\lambda\vdash N}
\Spe{\lambda}\times\Spe{\lambda}
\end{equation}
(see, for example, \cite[Claim 3.1]{Ros14}, \cite{amrr:index}).
For a transposition $\pi=(x,x')$, we often write $V_{xx'}$ instead of $V_{(x,x')}$. 

It is also convenient to consider $V$ as a representation of just $\Sym_\Dom$ or just $\Sym_\Ran$, and, respectively, we may write $V_\pi:=V_\pi^\epsilon$ and $V^\tau:=V^\tau_\epsilon$, where $\epsilon$ denote the respective identity permutations.
In such a case, $V$ contains $\dm{\lambda}$ instances of the irrep $\Spe{\lambda}$.%
\footnote{If we interpret $\Sym_\Dom=\Func=\Sym_\Ran$, then $V_\pi$ and $V^\tau$ are, respectively, the left and the right regular representations of the group.} 

\subsection{Notation Regarding Various Subspaces}

For $\lambda\vdash N$, $\mu\vdash N-1$, $\nu\vdash N-2$ and distinct $y,y'\in\Ran$, let $\cH^\lambda$, $\cH^{\mu_y}$, and $\cH^{\nu_{yy'}}$ be the isotypic subspaces of $\cF$ corresponding to the irrep $\Spe{\lambda}$ of $\Sym_\Ran$, the irrep $\Spe{\mu}$ of $\Sym_{\Ran\setminus\{y\}}$, and the irrep $\Spe{\nu}$ of $\Sym_{\Ran\setminus\{y,y'\}}$.
Similarly, for $\lambda'\vdash N$, $\mu'\vdash N-1$, and $x\in\Dom$, let $\cH_{\lambda'}$ and $\cH_{\mu'_x}$ be the isotypic subspaces corresponding to, respectively, the irrep $\Spe{\lambda'}$ of $\Sym_\Dom$ and the irrep $\Spe{\mu'}$ of $\Sym_{\Dom\setminus\{x\}}$.

As per our convention on the notation, let $\Pi^\lambda$, $\Pi^{\mu_y}$, $\Pi^{\nu_{yy'}}$, $\Pi_{\lambda'}$, $\Pi_{\mu'_x}$  denote the orthogonal projectors on $\cH^\lambda$, $\cH^{\mu_y}$, $\cH^{\nu_{yy'}}$, $\cH_{\lambda'}$, $\cH_{\mu'_x}$, respectively.
In this paper, we mostly use subscripts for irreps of $\Sym_\Dom$ and its subgroups, superscripts for irreps of $\Sym_\Ran$ and its subgroups, and overscripts for both.

The group actions (\ref{eqn:ActOnF}) of $\Sym_\Dom$ and  $\Sym_\Ran$ commute, therefore all $\Pi^\lambda$, $\Pi^{\mu_y}$, $\Pi^{\nu_{yy'}}$, $\Pi_{\lambda'}$, $\Pi_{\mu'_x}$ also mutually commute.
However, for $\mu,\mu'\vdash N-1$ and for distinct $y,y'\in\Ran$, $\Pi^{\mu_y}$ and $\Pi^{\mu'_{y'}}$ in general do not commute.
When two or more such projectors do commute, for the sake of brevity, we often denote their product by combining their scripts. For example,
\[
\Pi_\lambda^{\mu_y,\nu_{yy'}}
:=
\Pi_\lambda
\Pi^{\mu_y}\Pi^{\nu_{yy'}}.
\]

\paragraph{Overscripts.}

Let $x\in \Dom$ and $y\in\Ran$.
Using the branching rule, one can show that $\Pi_{\mu_x}\ne \Pi^{\mu_y}$ for any $\mu\vdash N-1$ (we do not use this fact).
On the other hand, for $\lambda\vdash N$, 
due to the decomposition (\ref{eqn:LambdaLambdaDecom}), we have
\[
\Pi_\lambda = \Pi^\lambda = \Pi_\lambda^\lambda
\]
and we denote this projector by $\Piooo{\lambda}{}{}$.
This notation lets us ``bring down'' the overscripts as superscripts or subscripts, if we want to consider $\Spe{\lambda}$ as a representation of $\Sym_\Ran$ or $\Sym_\Dom$, respectively.
For example, by the branching rule, the projectors $\Piooo{\lambda}{\mu_y}{}$ and $\Piooo{\lambda}{}{\mu_x}$, where $\mu\vdash N-1$, are non-zero if and only if $\mu\yunder\lambda$.
Note that all three $\Piooo{\lambda}{\mu_y}{}$, $\Piooo{}{\lambda,\mu_y}{}$, and $\Piooo{}{\mu_y}{\lambda}$ denote the same projector, and the notation we choose to use may depend on the desired emphasis, as well as the fact that the latter two occupy less vertical space.
We have
\begin{equation}
\label{eqn:IFDec}
I_\regF=\sum_{\lambda\vdash N}\Piooo{\lambda}{}{}
\qquad\text{and}\qquad
\cF=\bigoplus_{\lambda\vdash N}\Hooo{\lambda}{}{}.
\end{equation}

\paragraph{Bijections on $N-1$ elements.}

Recall that $\Xi_x^y$ projects on all bijections from $\Dom$ to $\Ran$ that map $x$ to $y$.
These bijections, in turn, are trivially in one-to-one correspondence to all bijections from $\Dom\setminus\{x\}$ to $\Ran\setminus\{y\}$.
Hence, analogously to (\ref{eqn:IFDec}), we can write
\begin{equation}\label{eqn:XiBasicDecomp}
\Xi_x^y = \sum_{\mu\vdash N-1}\Xiooo{\mu}{x}{y},
\end{equation}
where
\[
\Xiooo{\mu}{y}{x} := \Xi_x^y\Pi^{\mu_y} = \Xi_x^y\Pi_{\mu_x}
\]
is an orthogonal projector on a single instance of the irrep $\Spe{\mu}\otimes\Spe{\mu}$ of $\Sym_{\Dom\setminus\{x\}}\times\Sym_{\Ran\setminus\{y\}}$.

\subsection{Definitions of $\spA{k}$ and $\spB{k}$ via Representation Theory}
\label{sec:spAspBviaReps}

Let us first consider $\spA{k}$ as a subspace of $\cF$. 
Note that $\spA{k}$ is invariant under the action of $\Sym_\Dom\times\Sym_\Ran$, and thus it is a representation of the group.
Recall the decomposition~(\ref{eqn:XiBasicDecomp}) of $\cF$, in which
the space $\Hooo{\lambda}{}{}$ is the $(\Spe{\lambda}\otimes\Spe{\lambda})$-isotypic subspace containing the sole instance of the irrep.
Hence, for every $\lambda\vdash N$, we either have $\Hooo{\lambda}{}{}\subseteq\spA{k}$ or $\Hooo{\lambda}{}{}\perp\spA{k}$.

Now consider $\spB{k}\subseteq\cF$, which is invariant under the action of $\Sym_\Dom\times\Sym_{\Ran\setminus\{\zero\}}$ and, thus, is a representation of the group.
The space $\cF$ considered as a representation of $\Sym_\Dom\times\Sym_{\Ran\setminus\{\zero\}}$ decomposes into irreps as 
\[
\cF=\bigoplus_{\substack{\lambda\vdash N,\,\mu\vdash N-1\\ \mu\yunder\lambda}}\Hooo{\lambda}{\mu_\zero}{},
\]
and the space $\Hooo{\lambda}{\mu_\zero}{}$ is the $(\Spe{\lambda}\otimes\Spe{\mu})$-isotypic subspace containing the sole instance of the irrep.
Hence, for every $\lambda\vdash N$ and $\mu\vdash N-1$ such that $\mu\yunder\lambda$, we either have $\Hooo{\lambda}{\mu_\zero}{}\subseteq\spB{k}$ or $\Hooo{\lambda}{\mu_\zero}{}\perp\spB{k}$.

\begin{clm}
\label{clm:spABviaReps}
We have $\spA{k}=\bigoplus_{\theta\in\lightPart{k}}\Hooo{\bar\theta}{}{}$ and 
$\spB{k+1}=\bigoplus_{\theta\in\lightPart{k}}\cH^{\bar\theta_\zero}$.
\end{clm}

For example, for small values of $k$, the above claim states that
\begin{align*}
& \spA{0}=\underline{\cH_{(N)}^{(N-1)_\zero}}, \\
& \spB{1}=\cH_{(N)}^{(N-1)_\zero} \oplus \underline{\cH_{(N-1,1)}^{(N-1)_\zero}}, \\
& \spA{1}=\cH_{(N)}^{(N-1)_\zero} \oplus \cH_{(N-1,1)}^{(N-1)_\zero}\oplus \underline{\cH_{(N-1,1)}^{(N-2,1)_\zero}}, \\
& \spB{2}=\cH_{(N)}^{(N-1)_\zero} \oplus \cH_{(N-1,1)}^{(N-1)_\zero}\oplus \cH_{(N-1,1)}^{(N-2,1)_\zero}\oplus \underline{\cH_{(N-2,2)}^{(N-2,1)_\zero} \oplus \cH_{(N-2,1,1)}^{(N-2,1)_\zero}}, \\
& \spA{2}=\cH_{(N)}^{(N-1)_\zero} \oplus \cH_{(N-1,1)}^{(N-1)_\zero}\oplus \cH_{(N-1,1)}^{(N-2,1)_\zero}\oplus \cH_{(N-2,2)}^{(N-2,1)_\zero} \oplus \cH_{(N-2,1,1)}^{(N-2,1)_\zero}\oplus \underline{\cH_{(N-2,2)}^{(N-3,2)_\zero} \oplus \cH_{(N-2,1,1)}^{(N-3,1,1)_\zero}}.
\end{align*}
Recalling the definitions $\overline\cH_k^\hi := \spB{k}\cap(\spA{k-1})^\perp$ and $\overline\cH_k^\lo := \spA{k}\cap(\spB{k})^\perp$, 
the underlined spaces are $\overline\cH_0^\lo$, $\overline\cH_1^\hi$, $\overline\cH_1^\lo$, $\overline\cH_2^\hi$, $\overline\cH_2^\lo$, respectively.
In general, due to the above claim, we decompose $\overline\cH_k^\lo$ and $\overline\cH_k^\hi$ into irreps of $\Sym_\Dom\times\Sym_{\Ran\setminus\{\zero\}}$ as follows.
\begin{cor}
\label{cor:HlohiViaIrreps}
We have
\[
\overline\cH^\lo_{k} = 
\bigoplus_{\theta\vdash k}
\Hooo{\bar\theta}{\bar\theta_\zero}{}
\qquad\text{and}\qquad
\overline\cH^\hi_{k} = 
\bigoplus_{\substack{\theta\vdash k,\,\rho\vdash k-1\\\rho\yunder\theta}}
\Hooo{\bar\theta}{\bar\rho_\zero}{}.
\]
\end{cor}

\begin{proof}[Proof of Claim~\ref{clm:spABviaReps}]
Let us first address the decomposition of $\spA{k}$ into irreps.
Recall that for each $\lambda\vdash N$
we have either $\Hooo{\lambda}{}{}\subseteq\spA{k}$ or $\Hooo{\lambda}{}{}\perp\spA{k}$.

\paragraph{Proof of $\spA{k}\supseteq\bigoplus_{\theta\in \lightPart{k}}\cH_{\bar\theta}$.}
Given $\theta\vdash k'$ with $k'\le k$ and $f\in\Func$, it suffices to show that $\Pi_{\bar\theta}|f\>\in\spA{k}$.
Recall Eqn.~(\ref{eqn:PiRhoAlt}) that expresses $\Pi_{\bar\theta}$ via Young tableaux $\tabl$.
Given $\tabl$ of shape $\bar\theta$, it suffices to show that $P_\tabl|f\>\in\spA{k}$.
Further on, since $P_\tabl=P_\tabl \widehat{E}_{R_1(\tabl)}^+$,
 it suffices to prove $\widehat{E}_{R_1(\tabl)}^+|f\>\in\spA{k}$,
because $\spA{k}$ being a representation of $\Sym_\Dom$ implies that it is invariant under $P_\tabl$.

Let $\{x_1,\ldots,x_{k'}\}:=\Dom\setminus R_1(\tabl)$ and let  $y_i:=f(x_i)$ for $i \in [k']$.
We note that 
\[
\widehat{E}^+_{R_1(\tabl)} |f\> = |\vn_{x_1,\ldots,x_{k'}}^{y_1,\ldots,y_{k'}}\> \big/\sqrt{(N-k)!}.
\]
As a result, $P_\tabl$ maps every $|f\>$ and, in turn, every vector in $\cF$ to a vector in $\spA{k'}\subseteq\spA{k}$. 

\paragraph{Proof of $\spA{k}\subseteq\bigoplus_{\theta\in \lightPart{k}}\cH_{\bar\theta}$.}
We have to show that $\spA{k}\perp\cH_{\bar\theta}$ for all $\theta$ such that $|\theta|\ge k+1$, and we do that by showing that $\cH_{\bar\theta}$ maps vectors spanning $\spA{k}$ to $0$.
Take a vector $|\vn_{x_1,\ldots,x_k}^{y_1,\ldots,y_k}\> $ and any Young tableau $\tabl$ of shape $\bar\theta$, where $|\theta|\ge k+1$. One can see that, because $|\theta|\ge k+1$, there exists a column $C_j(\tabl)$ of $\tabl$ that contains two elements 
$x',x''$ outside $\{x_1,\ldots,x_k\}$.
For this column, we have 
\[
\widehat{E}^-_{C_j(\tabl)}=\widehat{E}^-_{C_j(\tabl)} \widehat{E}^-_{\{x',x''\}}
\]
and, in turn,
\[
P_\tabl^\top = P_\tabl^\top \widehat{E}^-_{\{x',x''\}}.
\]
However, since $x'$ and $x''$ are outside $\{x_1,\ldots,x_k\}$, the vector $|\vn_{x_1,\ldots,x_k}^{y_1,\ldots,y_k}\> $ is fixed by $V_{x'x''}$, and thus it is nullified by
$\widehat{E}^-_{\{x',x''\}}=(I-V_{x'x''})/2$ and, in turn, $P_\tabl^\top$.
Since this holds for any tableau $\tabl$, we have $\Pi_{\bar\theta} |\vn_{x_1,\ldots,x_k}^{y_1,\ldots,y_k}\>  = 0$.

\paragraph{Proof of $\spB{k}=\bigoplus_{\theta\in \lightPart{k}}\cH^{\bar\theta_\zero}$.}
We can write $\cF=\oplus_{x\in\Dom}\cF_x^\zero$, where $\cF_x^\zero:=\{|f\>\colon f(x)=\zero\}$ is the image of $\Xi_x^\zero$.
We can write
\[
\spB{k}=
\bigoplus_{x\in\Dom}
\spB{k,x},
\qquad\text{where}\qquad
\spB{k,x}:=\spn\big\{|\vn_{x,x_2,\ldots,x_{k}}^{\zero,y_2,\ldots,y_{k}}\>\colon x_i\in \Dom\setminus\{x\}\And y_i\in\Ran\setminus\{\zero\}\big\}\subseteq\cF_x^0.
\]
The rest of the proof proceeds exactly as for $\spA{k}$, except now, instead of permutations over $\Dom$ (or, equivalently, over $\Ran$), we consider permutations over $\Ran\setminus\{0\}$.
\end{proof}

\section{Proof of Theorem~\ref{thm:mainThm} for the Permutation Case}
\label{sec:MainTmProof}

In this section, we prove Theorem~\ref{thm:mainThm}. While originally $\cH_k^\lo$ and $\cH_k^\hi$ were defined via vectors $|\vn_{x_1,\ldots,x_k}^{y_1,\ldots,y_k}\>$, in the proof we only use their expression via representation theory arising from Corollary~\ref{cor:HlohiViaIrreps}.

\paragraph{Alternative proof of Point 1.}

While we have already proven Point 1 in Section~\ref{sec:MainTmFrame}, here we provide an alternative proof based on the representation theory.
Recall the expression~(\ref{eq:ODecomp}) that expresses $O$ via projectors
\(
\Xi_x^y = \sum_{\mu\vdash N-1} \Xiooo{\mu}{y}{x}.
\)
By the branching rule, we have
\begin{equation}\label{eqn:LambdaDiff}
\Piooo{\lambda}{}{}\Xi_x^y \Piooo{\lambda'}{}{}
= \sum_{\substack{\mu\vdash N-1\\\mu\yunder\lambda,\,\mu\yunder\lambda'}}\Piooo{\lambda}{}{}\Xiooo{\mu}{y}{x} \Piooo{\lambda'}{}{}.
\end{equation}
Hence, if $\lambda=\bar\theta$ and $\lambda'=\bar\eta$ for $\eta,\theta$ such that $\big||\theta|-|\eta|\big|>1$, no $\mu$ exists satisfying both $\mu\yunder\lambda$ and $\mu\yunder\lambda'$, therefore $\Piooo{\bar\theta}{}{}\Xi_x^y \Piooo{\bar\eta}{}{}=0$.
This concludes the proof of Point~1 of Theorem~\ref{thm:mainThm}.

\paragraph{Proof of Point 2.}

Recall that we have already shown in Section~\ref{sec:MainTmFrame} that
\[
\big\|P|\psi_k\>\big\|
\le  \alpha_k + \big\|\Xi_x^\zero\Pi^\lo_k\big\|,
\]
and it is left to bound the latter norm.
Since $I_{\regF}=\sum_{\mu\vdash N-1}\Pi^{\mu_\zero}$
and $\Pi^\lo_k$ and $\Xi_x^\zero$ commute with all $\Pi^{\mu_\zero}$, by Property~5 of Fact~\ref{fac:norms} we have
\[
\big\|\Xi_x^\zero\Pi^\lo_k\big\|
=
\max_{\theta\in\lightPart{k}}
\Big\|\Xi_x^\zero\Piooo{\bar\theta}{\bar\theta_\zero}{}\Big\|.
\]
To evaluate the norm under maximization, more generally we have
\begin{multline*}
\Big\|\Xiooo{\mu}{\zero}{x}\Piooo{\lambda}{\mu_\zero}{}\Big\|^2
=\Big\|\Xiooo{\mu}{0}{x}\Piooo{\lambda}{\mu_0}{}\Xiooo{\mu}{0}{x}\Big\|
=\frac{\tr\Big[\Xiooo{\mu}{0}{x}\Piooo{\lambda}{\mu_0}{}\Big]}{\dm\mu^2}
=\frac{\tr\Big[\Xiooo{}{\zero}{x}\Piooo{\lambda}{\mu_0}{}\Big]}{\dm\mu^2}
\\=\frac{\sum_{x\in\Dom}\tr\Big[\Xiooo{}{0}{x}\Piooo{\lambda}{\mu_0}{}\Big]}{N\dm{\mu}^2}
=\frac{\tr\Big[\Piooo{\lambda}{\mu_0}{}\Big]}{N\dm{\mu}^2}
=\frac{\dm\lambda}{N\dm\mu}.
\end{multline*}
The proof of Point 2 of Theorem~\ref{thm:mainThm} concludes by taking $\lambda=\bar\theta$ and $\mu=\bar\theta_*$ and applying the following claim, which is also used in the proof of Point~3.
\begin{clm}[{\cite[Claim~6.3]{Ros14}}]\label{clm:dimClose}
For $\theta\vdash k$, we have
$
\frac{\dim\bar\theta}{\dim\bar\theta_*} \le \frac{N}{N-2k}.
$
\end{clm}

\paragraph{Proof of Point 3.}

Recall that we have already shown in Section~\ref{sec:MainTmFrame} that
\[
\alpha_k 
\le \alpha_{k-1}
+
2\Big(
\sum_{y} \big\|\Pi^\lo_k \Xi_x^y \Pi^\hi_k \big\|^2
\Big)^{1/2}.
\]

\begin{clm}\label{clm:Yis0transition}
We have $\big\|\Pi^\lo_k \Xi_x^0 \Pi^\hi_k \big\|\le 1/\sqrt{N-2k}$.
\end{clm}

\begin{clm}\label{clm:Ynot0transition}
We have $\big\|\Pi^\lo_k \Xi_x^y \Pi^\hi_k \big\|\le 1/(N-2k)$ for any $y\in\Ran\setminus\{0\}$.
\end{clm}

By a simple observation that $(N-1)/(N-2k)^2 < 1/(N-4k)$, we get that
\[
2\Big(
\sum_{y} \big\|\Pi^\lo_k \Xi_x^y \Pi^\hi_k \big\|^2
\Big)^{1/2}
<
\frac{2\sqrt{2}}{\sqrt{N-4k}}.
\]

\begin{proof}[Proof of Claim~\ref{clm:Yis0transition}]
Note that all $\Pi^\lo_k$, $\Xi_x^0$, and $\Pi^\hi_k$ commute with $\Pi^{\mu_0}$ for any $\mu\vdash N-1$. Hence, by Property~5 of Fact~\ref{fac:norms},
\[
\big\|\Pi^\lo_k \Xi_x^0 \Pi^\hi_k \big\|
=
\max_{\theta\in\lightPart{k-1}}
\bigg\|\Piooo{\bar\theta}{\bar\theta_0}{} \Xi_x^0 \sum_{\chi\succ\theta}\Piooo{\bar\chi}{\bar\theta_0}{} \bigg\|
=
\max_{\theta\in\lightPart{k-1}}
\bigg\|
\Piooo{\bar\theta}{\bar\theta_0}{} \Xi_x^0 \sum_{\chi\succ\theta}\Piooo{\bar\chi}{\bar\theta_0}{}
\Xi_x^0 \Piooo{\bar\theta}{\bar\theta_0}{}
\bigg\|^{1/2}
\]

So, for a given $\theta$, we have
\begin{multline*}
\left\|
\Piooo{\bar\theta}{\bar\theta_0}{\bar\theta_x}
\Xiooo{\bar\theta_*}{0}{x}
\sum_{\chi\succ\theta}
\Piooo{\bar\chi}{\bar\theta_0}{\bar\theta_x}
\Xiooo{\bar\theta_*}{0}{x}
\Piooo{\bar\theta}{\bar\theta_0}{\bar\theta_x}
\right\|
=
\frac{
\sum_{\chi\succ\theta}
\tr\left[
\Piooo{\bar\theta}{\bar\theta_0}{\bar\theta_x}
\Xiooo{}{0}{x}
\Piooo{\bar\chi}{\bar\theta_0}{\bar\theta_x}
\Xiooo{}{0}{x}
\right]
}{(\dim\bar\theta_*)^2}
=
\frac{
\sum_{\chi\succ\theta}
\tr\left[
\Piooo{\bar\theta}{\bar\theta_0}{\bar\theta_x}
\Xiooo{}{0}{x}
\right]
\tr\left[
\Piooo{\bar\chi}{\bar\theta_0}{\bar\theta_x}
\Xiooo{}{0}{x}
\right]
}{(\dim\bar\theta_*)^4}
\\
=
\frac{
\dim\bar\theta
\sum_{\chi\succ\theta}
\dim\bar\chi
}{N^2(\dim\bar\theta_*)^2}
<
\frac{\dim\bar\theta}{N \dim\bar\theta_*}
\le \frac{1}{N-2k+2}
\end{multline*}
where the second equality is due to Lemma~\ref{lem:isotypicOverlap} and where the two inequalities are due to
\begin{equation}\label{eqn:IndBranchIneq}
\sum_{\chi\succ\theta}\dim\bar\chi = N\dim\bar\theta_* - \dim\bar\theta < N\dim\bar\theta_*
\end{equation}
as a result of the branching rule and Claim~\ref{clm:dimClose}.
\end{proof}

\begin{proof}[Proof of Claim~\ref{clm:Ynot0transition}]

Note that $\Pi^\lo_k$, $\Xi_x^y$, and $\Pi^\hi_k$ all commute with $\Pi_{\mu_x}^{\nu_{0y}}$ for any $\mu\vdash N-1$ and $\nu\vdash N-2$
and that $\sum_{\mu,\nu}\Pi_{\mu_x}^{\nu_{0y}}=I$. Hence,
\begin{align*}
\big\|\Pi^\lo_k \Xi_x^y \Pi^\hi_k \big\|
\,&=
\max\bigg\{
\max_{\theta\in\lightPart{k}}
\Big\|\Pi^\lo_k \Xiooo{\bar\theta_*}{y,\bar\theta_{0y}}{x} \Pi^\hi_k \Big\|,
\,
\max_{\substack{\theta\in\lightPart{k}\\\rho\yunder\theta}}
\Big\|\Pi^\lo_k \Xiooo{\bar\theta_*}{y,\bar\rho_{0y}}{x} \Pi^\hi_k \Big\|
\bigg\}
\\&\le
\max\bigg\{
\max_{\theta\in\lightPart{k}}
\Big\|\Xiooo{\bar\theta_*}{y,\bar\theta_{0y}}{x} \Pi^\hi_k
\Big\|,
\,
\max_{\substack{\theta\in\lightPart{k}\\\rho\yunder\theta}}
\Big\|\Pi^\lo_k \Xiooo{\bar\theta_*}{y,\bar\rho_{0y}}{x}\Big\|
\bigg\}.
\end{align*}
We will bound the two norms separately. Both bounds use the following claim.

\begin{clm}\label{clm:boxDist}
Suppose $\theta\vdash k$ and $\rho\vdash k-1$ satisfy $\rho\yunder\theta$. Let
$\mu\vdash N-1$ be such that $\bar\theta\yover\mu\yover\bar\rho_{**}$, namely 
either $\mu=\bar\theta_*$ or $\mu=\bar\rho_*$.
Then
\[
\tr\Big[
\Piooo{\bar\theta}{\mu_0,\bar\rho_{0y}}{}
\Piooo{\bar\theta}{\mu_y,\bar\rho_{0y}}{}
\Big]
\le \frac{\dim\bar\theta \dim\bar\rho_{**}}{(N-2k+1)^2}.
\]
\end{clm}

\begin{proof}
Note that $V^{0y}$ commutes with $\Piooo{}{\bar\rho_{0y}}{\bar\theta}$ and that
$\Piooo{}{\mu_y,\bar\rho_{0y}}{\bar\theta}=V^{0y}\Piooo{}{\mu_0,\bar\rho_{0y}}{\bar\theta}
V^{0y}$.
Eqn.~(5.6) of \cite{Ros14} expresses the restriction of $V^{0y}$ to $\Hooo{}{\bar\rho_{0y}}{\bar\theta}$ as a liner combination of projectors on irreps $\Hooo{}{\bar\theta_0,\bar\rho_{0y}}{\bar\theta}$ and 
$\Hooo{}{\bar\rho_0,\bar\rho_{0y}}{\bar\theta}$ as well as isomorphisms between them.
A direct consequence of that expression is that
\[
\tr\Big[
\Piooo{}{\mu_0,\bar\rho_{0y}}{\bar\theta}
\cdot
V^{0y}
\Piooo{}{\mu_0,\bar\rho_{0y}}{\bar\theta}
V^{0y}
\Big]
=
\frac{\dim\bar\theta \dim\bar\rho_{**}}{\mathfrak{D}(\bar\theta,\bar\rho_{**})^2}
\]
where $\mathfrak{D}(\bar\theta,\bar\rho_{**})$ is the distance between two boxes that have to be removed from the Young diagram $\bar\theta$ to obtain $\bar\rho_{**}$.
The proof concludes by observing that $\mathfrak{D}(\bar\theta,\bar\rho_{**})\ge N-2k+1$.
\end{proof}

Let us first bound $\big\|\Pi^\lo_k \Xiooo{\bar\theta_*}{y,\bar\rho_{0y}}{x}\big\|$. Notice that
\[
\Xiooo{\bar\theta_*}{y,\bar\rho_{0y}}{x} 
=
\Big(
\Piooo{\bar\theta}{\bar\rho_{0}}{}
+\Piooo{\bar\theta}{\bar\theta_{0}}{}
+ \sum_{\substack{\chi\yover\theta,\,\eta\yover\rho
}}
\Piooo{\bar\chi}{\bar\eta_{0}}{}
\Big)
\Xiooo{\bar\theta_*}{y,\bar\rho_{0y}}{x}.
\]
Since $\chi\yover\theta$ and $\eta\yover\rho$ together imply that $|\chi|=|\eta|+1$ and, in turn, that $\Pi^\lo_k$ and $\Piooo{\bar\chi}{\bar\eta_{0}}{}$ are orthogonal,
we get that 
$\Pi^\lo_k \Xiooo{\bar\theta_*}{y,\bar\rho_{0y}}{x}$ equals $\Piooo{\bar\theta}{\bar\theta_{0}}{} \Xiooo{\bar\theta_*}{y,\bar\rho_{0y}}{x}$
when $|\theta|\le k$ and $0$ otherwise. Let us assume $|\theta|\le k$.
We have
\begin{multline*}
\left\|
\Piooo{\bar\theta}{\bar\theta_0}{}
\Xiooo{\bar\theta_*}{y,\bar\rho_{0y}}{x}
\right\|^2
=
\left\|
\Xiooo{}{y,\bar\rho_{0y}}{x,\bar\theta_x}
\Piooo{\bar\theta}{\bar\theta_0,\bar\rho_{0y}}{\bar\theta_x}
\Xiooo{}{y,\bar\rho_{0y}}{x,\bar\theta_x}
\right\|
=
\frac{
\tr\left[
\Xiooo{}{y,\bar\rho_{0y}}{x,\bar\theta_x}
\Piooo{\bar\theta}{\bar\theta_0,\bar\rho_{0y}}{\bar\theta_x}
\Xiooo{}{y,\bar\rho_{0y}}{x,\bar\theta_x}
\right]
}
{\dim\bar\theta_*\dim\bar\rho_{**}}
\\
=
\frac{
\tr\left[
\Xi_x^y
\Pi^{\bar\theta_y}
\Piooo{\bar\theta}{\bar\theta_0,\bar\rho_{0y}}{}
\right]
}
{\dim\bar\theta_*\dim\bar\rho_{**}}
=
\frac{
\tr\left[
\Piooo{\bar\theta}{\bar\theta_y,\bar\rho_{0y}}{}
\Piooo{\bar\theta}{\bar\theta_0,\bar\rho_{0y}}{}
\right]
}
{N\dim\bar\theta_*\dim\bar\rho_{**}}
<
\frac{\dim\bar\theta}
{(N-2k)^2 N\dim\bar\theta_*}
<
\frac{1}{(N-2k)^{3}}
\end{multline*}
where the two inequalities are due to Claims \ref{clm:boxDist} and \ref{clm:dimClose}.

Let us now bound $\big\|
\Pi^{\hi}_k\Xiooo{\bar\theta_*}{y,\bar\theta_{0y}}{x}
\big\|$.
Notice that
\[
\Xiooo{\bar\theta_*}{y,\bar\theta_{0y}}{x}
=
\Big(
\Piooo{\bar\theta}{\bar\theta_0}{}
+
\sum_{\chi\yover\theta}
\big(
\Piooo{\bar\chi}{\bar\theta_0}{} + \Piooo{\bar\chi}{\bar\chi_0}{}
\big)
\Big)
\Xiooo{\bar\theta_*}{y,\bar\theta_{0y}}{x}.
\]
As a result,
$\Pi^{\hi}_k\Xiooo{\bar\theta_*}{y,\bar\theta_{0y}}{x}$
equals
$\sum_{\chi\yover\theta}
\Piooo{\bar\chi}{\bar\theta_0}{}
\Xiooo{\bar\theta_*}{y,\bar\theta_{0y}}{x}
$
when $|\theta|\le k-1$ and $0$ otherwise. Let us assume $|\theta|\le k-1$.
We have
\begin{multline*}
\Big\|
\sum_{\chi\yover\theta}
\Piooo{\bar\chi}{\bar\theta_0}{}
\Xiooo{\bar\theta_*}{y,\bar\theta_{0y}}{x}
\Big\|^2
=
\Big\|
\Xiooo{\bar\theta_*}{y,\bar\theta_{0y}}{x}
\sum_{\chi\yover\theta}
\Piooo{\bar\chi}{\bar\theta_0}{}
\Xiooo{\bar\theta_*}{y,\bar\theta_{0y}}{x}
\Big\|
=
\frac{
\tr\left[
\sum_{\chi\succ\theta}
\Piooo{\bar\chi}{\bar\theta_0}{}
\Piooo{\bar\chi}{\bar\theta_y,\bar\theta_{0y}}{}
\right]
}
{N\dim\bar\theta_*\dim\bar\theta_{**}} 
\\
<
\frac{
\sum_{\chi\succ\theta}
\dim\bar\chi
}
{(N-2k)^2 N \dim\bar\theta_*}
<\frac1{(N-2k)^2},
\end{multline*}
where the two inequalities are due to Claim~\ref{clm:boxDist} and the branching rule, in particular, Enq.~(\ref{eqn:IndBranchIneq}). 
\end{proof}

\section*{Acknowledgements}

The author would like to thank Fran\c{c}ois Le~Gall, Akinori Hosoyamada, and Tetsu Iwata for fruitful discussions.
Part of this work was done while the author was a JSPS International Research Fellow supported by the JSPS KAKENHI Grant Number JP19F19079. This work was supported by JSPS KAKENHI Grant Number JP20H05966 and MEXT Quantum Leap Flagship Program (MEXT Q-LEAP) Grant Number JPMXS0120319794.

{
\small
\newcommand{\etalchar}[1]{$^{#1}$}

}

\appendix

\section{Overlap of \texorpdfstring{$|\vn_{x_1,\ldots,x_k}^{y_1,\ldots,y_k}\>$}{Vectors v} on Various Subspaces}
\label{app:LargeOlaps}

In this section we prove Claims~\ref{clm:LargeOlap0} and \ref{clm:LargeOlap1}, which address what are the overlaps of vectors $|\vn_{x_1,\ldots,x_k}^{y_1,\ldots,y_k}\>$ on some specific subspaces.
When proving Claim~\ref{clm:LargeOlap0}, we prove an even stronger statement giving a concise expression for the overlap of $|\vn_{x_1,\ldots,x_k}^{y_1,\ldots,y_k}\>$ on any $\cH_\lambda$.

Recall that there are $\binom{N}{k}\frac{N!}{(N-k)!}$ assignments of $\alpha$ weight $k$.
Define the operator 
\[
\MOp_k:=(N-k)!\sum_{\alpha\colon|\alpha|=k}|\vn_\alpha\>\<\vn_\alpha|
\]
on $\cF$, whose support is clearly $\spA{k}$.
Note that $\MOp_k$ is fixed by $\Sym_\Dom\times\Sym_\Ran$, and therefore we can write
\[
\MOp_k= 
\sum_{\lambda\vdash N}
\mathrm{eval}_{\lambda}(k) \Pi_\lambda.
\]
The following lemma provides an expression for the eigenvalues $\mathrm{eval}_{\lambda}(k)$.
Let $h_1(\lambda)$ denote the product of the hook lengths of boxes in the first row of $\lambda$ and let $\lambda_{\geq2}$ denote the diagram formed by the boxes of $\lambda$ below the first row.
Recall that $|\lambda_{\geq2}|$ is the number of boxes in $\lambda_{\ge 2}$.

\begin{lem}\label{lem:IJ_perm_evals}
Suppose $k<N/2$.
 The eigenvalue of $\MOp_k$ corresponding to $\lambda\vdash N$ is 
\(
h_1(\lambda)/(k-|\lambda_{\geq2}|)!
\)
when $|\lambda_{\geq2}|\leq k$, and $0$ otherwise.
\end{lem}

\noindent
We prove the lemma in Appendix~\ref{sec:IJ_perm_evals}.

\subsection{Proof of Claim~\ref{clm:LargeOlap0}}

We have
\[
\MOp_k=(N-k)!\sum_{\alpha\colon|\alpha|=k}|\vn_\alpha\>\<\vn_\alpha| = 
\sum_{\theta\in\lightPart{k}}
\mathrm{eval}_{\bar\theta}(k) \Pi_{\bar\theta}.
\]
By multiplying this from both left and right with $\Pi_{\bar\theta}$, where $\theta\in\lightPart{k}$, we get
\[
(N-k)!\sum_{\alpha\colon|\alpha|=k}\Pi_{\bar\theta}|\vn_\alpha\>\<\vn_\alpha|\Pi_{\bar\theta} = \mathrm{eval}_{\bar\theta}(k) \Pi_{\bar\theta},
\]
whose trace is
\[
(N-k)!\sum_{\alpha\colon|\alpha|=k}
\|\Pi_{\bar\theta}|\vn_\alpha\>\|^2
= \mathrm{eval}_{\bar\theta}(k) d_{\bar\theta}^2,
\]
which equals $\binom{N}{k}N!\|\Pi_{\bar\theta}|\vn_\alpha\>\|^2$ due to the symmetry among all $\alpha$.
Hence,
\[
\|\Pi_{\bar\theta}|\vn_\alpha\>\|^2
= \frac{d_{\bar\theta}^2 \mathrm{eval}_{\bar\theta}(k)}{\binom{N}{k}N!}
= \frac{d_{\bar\theta}^2 h_1(\bar\theta)}{\binom{N}{k}N!(k-|\theta|)!}.
\]

By the hook length formulae for $d_{\bar\theta}$ and $d_\theta$, 
we have 
\begin{equation}
\label{eqn:hookLenExpand}
d_{\bar\theta}
= \frac{N!}{h(\bar\theta)}
= \frac{N!}{h_1(\bar\theta)\cdot h(\theta)}
= \frac{N!d_{\theta}}{h_1(\bar\theta)\cdot |\theta|!},
\end{equation}
from which we extract
\[
d_{\bar\theta}h_1(\bar\theta)\big/N!=d_\theta\big/|\theta|!.
\]
Hence,
\[
\|\Pi_{\bar\theta}|\vn_\alpha\>\|^2
= 
\frac{d_{\bar\theta}}{\binom{N}{k}}
\cdot
\frac{d_{\theta}}{(k-|\theta|)!|\theta|!},
\]
which gives us the amplitude of $|\vn_\alpha\>$ on the space $\cH_{\bar\theta}$.
Regarding Claim~\ref{clm:LargeOlap0}, we are interested in $\theta\vdash k$, as we have
\[
\spA{k}\cap(\spA{k-1})^\perp=\bigoplus_{\theta\vdash k}\cH_\theta.
\]
When $|\theta|=k$, we have
\begin{equation}
\label{eqn:AnyOlap0}
\|\Pi_{\bar\theta}|\vn_\alpha\>\|^2
= 
\frac{d_{\bar\theta}}{\binom{N}{k}}
\cdot
\frac{d_{\theta}}{k!},
\end{equation}
and using (\ref{eqn:hookLenExpand}) we get
\[
\|\Pi_{\bar\theta}|\vn_\alpha\>\|^2
= 
\frac{1}{\binom{N}{k}}
\cdot
\frac{N!d_\theta}{h_1(\bar\theta)\,k!}
\cdot
\frac{d_{\theta}}{k!}
=
\frac{d_\theta^2}{k!}\cdot\frac{(N-k)!}{h_1(\bar\theta)}
\in \frac{d_\theta^2}{k!}\cdot\bigg[\frac{(N-k)!}{\max_{\theta\vdash k}h_1(\bar\theta)},\,
\frac{(N-k)!}{\min_{\theta\vdash k}h_1(\bar\theta)}\bigg].
\]
Note that $\sum_{\theta \vdash k}d_\theta^2=k!$, 
and hence
\[
\bigg\|\Big(\sum_{\theta \vdash k}\Pi_{\bar\theta}\Big)|\vn_\alpha\>\bigg\|^2
=
\bigg[\frac{(N-k)!}{\max_{\theta\vdash k}h_1(\bar\theta)},\,
\frac{(N-k)!}{\min_{\theta\vdash k}h_1(\bar\theta)}\bigg].
\]
The proof of Claim~\ref{clm:LargeOlap1} concludes by the following claim.
\begin{clm}
Given $\theta\vdash k$, we have
\[
N(N-k-1)!= h_1\big((N-k,1^k)\big) \le h_1(\bar\theta) \le h_1\big((N-k,k)\big) = \frac{(N-k+1)!}{N-2k+1},
\]
where $(N-k,1^k)=(N-k,1,\ldots,1)$ with $1$ repeated $k$ times.
\end{clm}

\begin{proof}
The equalities of the claim follow easily from examining the hook lengths of $(N-k,1^k)$ and $(N-k,k)$. It remains to prove the inequalities.

Let $\theta,\theta'\vdash k$ be such that $\theta'$ is obtained from $\theta$ by moving the lowest box of column $j$ to the bottom of column $j'<j$. In the first rows of $\bar\theta$ and $\bar{\theta'}$, all but two boxes have the same hook lengths:
\[
h_{\bar{\theta'}}((1,\ell)) =
\begin{cases}
h_{\bar\theta}((1,\ell))+1 & \text{if }\ell = j';\\
h_{\bar\theta}((1,\ell))-1 & \text{if }\ell = j;\\
h_{\bar\theta}((1,\ell)) & \text{otherwise}.\\
\end{cases}
\]
Since $h_{\bar{\theta}}((1,j'))\ge h_{\bar{\theta}}((1,j))+1$, we have
\[
\frac{h_1(\bar{\theta'})}{h_1(\bar\theta)}
=
\frac{(h_{\bar\theta}((1,j'))+1)(h_{\bar\theta}((1,j))-1)}{h_{\bar\theta}((1,j'))h_{\bar\theta}((1,j))}
\le
1-
\frac{2}{h_{\bar\theta}((1,j'))h_{\bar\theta}((1,j))}
< 1.
\]
The aforementioned procedure for obtaining $\theta'$ from $\theta$ moves, essentially, boxes from right to left, with $(k)$ and $(1^k)$ being the initial and final extremes of such a procedure.
Thus the result follows by the induction.
\end{proof}

\subsection{Proof of Claim~\ref{clm:LargeOlap1}}

Let $\theta\vdash k$ and let we decompose $\cH_{\bar\theta}$ as
\[
\cH_{\bar\theta}=\cH^{\bar\theta_\zero}_{\bar\theta}
\oplus\bigoplus_{\rho\yunder\theta}\cH^{\bar\rho_\zero}_{\bar\theta},
\]
and all these spaces are subspaces of the $\mathrm{eval}_{\bar\theta}(k)$-eigenspace of $\MOp_k$.
Hence, similarly as before, we have
\[
(N-k)!\sum_{\alpha\colon|\alpha|=k}\Pi^{\bar\theta_\zero}_{\bar\theta}|\vn_\alpha\>\<\vn_\alpha|
\Pi^{\bar\theta_\zero}_{\bar\theta}
= \mathrm{eval}_{\bar\theta}(k) \Pi^{\bar\theta_\zero}_{\bar\theta},
\]
whose trace is
\[
(N-k)!
\sum_{\alpha\colon|\alpha|=k}
\|\Pi^{\bar\theta_\zero}_{\bar\theta}|\vn_\alpha\>\|^2
= \mathrm{eval}_{\bar\theta}(k) d_{\bar\theta}d_{\bar\theta_*}
\]
However, unlike before, now we have to consider two types of assignments $\alpha$.
The norm $\|\Pi^{\bar\theta_\zero}_{\bar\theta}|\vn_\alpha\>\|$ is $0$ if $\zero$ belongs to the image of $\alpha$ and it is the same for all other $\alpha$, of which there are $\binom{N-1}{k}\frac{N!}{(N-k)!}$.
Therefore
\[
\|\Pi^{\bar\theta_\zero}_{\bar\theta}|\vn_\alpha\>\|^2
= \frac{\mathrm{eval}_{\bar\theta}(k) d_{\bar\theta}d_{\bar\theta_*}}{\binom{N-1}{k}N!}
= \frac{d_{\bar\theta_*}}{\binom{N-1}{k}}
\cdot \frac{h_1(\bar\theta) d_{\bar\theta}}{N!}
= \frac{d_{\bar\theta_*}}{\binom{N-1}{k}}
\cdot \frac{d_{\theta}}{k!}.
\]
Note that this quantity is obtained if $N$ in (\ref{eqn:AnyOlap0}) is replaced by $N-1$ (and, thus, $\bar\theta$ by $\bar\theta_*$), which is why bounds in Claim~\ref{clm:LargeOlap0} are those in Claim~\ref{clm:LargeOlap1} with $N$ replaced by $N-1$.
The rest of the proof proceeds the same way as for Claim~\ref{clm:LargeOlap0}.

\subsection{Proof of Lemma \ref{lem:IJ_perm_evals}}
\label{sec:IJ_perm_evals}

The case $|\lambda_{\geq2}|>k$ holds because the support of $\MOp_k$ is $\spA{k}$.
 Thus, let us assume $|\lambda_{\geq2}|\leq k$.
 Without loss of generality, in the proof let $\Dom=\Ran=[N]$.

Let $\alpha$ be an assignment of weight $k$. 
We say that a permutation $f\colon\Dom\rightarrow\Dom$ agree with $\alpha$ if $f(x)=\alpha(x)$ for all $x$ in the domain of $\alpha$.
Let us represent operators acting on $\cF$ as $N!\times N!$ matrices written in the standard basis $\{|f\>\colon f\in\Func\}$ and let $A[\![f,h]\!]=\<f|A|h\>$ be the $(f,h)$-element of such a matrix $A$.
We have
\[
\big(|\vn_\alpha\>\<\vn_\alpha|\big)[\![f,h]\!] =
\begin{cases}
1/(N-k)! &\text{if both $f$ and $h$ agree with $\alpha$};\\
0&\text{otherwise}.
\end{cases}
\]
The number of assignments of weight $k$ agreeing with both $f$ and $h$ is
\[
\binom{|\{x\colon f(x)=h(x)\}|}{k} = \binom{\fxd(fh^{-1})}{k},
\]
where $\fxd(f)$ denotes the number of fixed points of $f$. Then
\[
\MOp_k[\![f,h]\!] = \binom{\fxd(fh^{-1})}{k}.
\]

 The eigenvalue of our interest is
\[
\mathrm{eval}_\lambda(k)
=\frac{\Tr\left[\MOp_k\Pi_\lambda\right]}{\Tr\left[\Pi_\lambda\right]}
= \frac{\Tr\left[\MOp_k\Pi_\lambda\right]}{\dm\lambda^2}.
\]
Note that $\Pi_\lambda$ projects on the $\Spe{\lambda}$-isotypic subspace of the (both, left and right) regular representation of $\Sym_N$, therefore we have
\[
\Pi_\lambda[\![f,h]\!] = \frac{\dm\lambda}{N!}\chi_\lambda(fh^{-1}),
\]
where $\chi_\lambda$ is the character corresponding to $\Spe{\lambda}$. Hence,
\begin{multline*}
\Tr\left[\MOp_k\Pi_\lambda\right]
= \sum_{f,h\in\mathbb{S}_N} \MOp_k[\![h,f]\!]\,\Pi_\lambda[\![f,h]\!]
= \sum_{f,h\in\mathbb{S}_N} \binom{\fxd(fh^{-1})}{k}\frac{\dm\lambda}{N!}\chi_\lambda(fh^{-1})
\\
= N!\dm\lambda \sum_{f\in\Sym_N} \frac{1}{N!}\binom{\fxd(f)}{k}\chi_\lambda(f)
= \frac{N!}{k!}\dm\lambda \left\langle X_k , \chi_\lambda \right\rangle,
\end{multline*}
where $X_k$ is the character of the representation corresponding to the action of $\Sym_N$ on $k$-permutations of $[N]$. Indeed, there are exactly $\binom{\fxd(f)}{k}k!$ $k$-permutations of $[N]$ that are fixed by $f\in\Sym_N$.

\begin{clm}\label{clm:kn_perm_char}
We have
\[
\left\langle X_k , \chi_\lambda \right\rangle
=
\binom{k}{|\lambda_{\geq2}|}\dm{\lambda_{\geq2}},
\]
where the dimension of the empty partition $\epsilon\vdash0$ is $1$.
\end{clm}

Putting everything together, we get
\begin{multline*}
\mathrm{eval}_\lambda(k)
= \frac{N!\dm{\lambda}\,\binom{k}{|\lambda_{\geq2}|}\dm{\lambda_{\geq2}}}{\dm{\lambda}^2k!}
= \frac{N!}{\dm{\lambda}}\,\frac{\dm{\lambda_{\geq2}}}{|\lambda_{\geq2}|!\,(k-|\lambda_{\geq2}|)!} 
\\
= h(\lambda)\,\frac{1}{h(\lambda_{\geq2})\,(k-|\lambda_{\geq2}|)!}
= \frac{h_1(\lambda)}{(k-|\lambda_{\geq2}|)!}.
\end{multline*}

\begin{proof}[Proof of Claim \ref{clm:kn_perm_char}] 
The group $\S_k\times\S_N$ acts on $k$-permutations of $[N]$ in the natural way. By inducing from $\Sym_k\times(\Sym_k\times\Sym_{N-k})$ using the special case (\ref{eqn:IndLR}) of the Littlewood--Richardson rule, we get that the corresponding representation decomposes into $\S_k\times\S_N$ irreps as
\[
\bigoplus_{\xi\vdash k}\bigoplus_{\theta\in\Lambda^{-}(\xi)} \Spe{\xi} \otimes \Spe{(N-|\theta|,\theta)},
\]
where $\Lambda^{-}(\xi)$ is the set of all Young diagrams that can be obtained from $\xi$ by removing at most one box per each column. Think of $\theta$ as $\lambda_{\geq2}$. Hence, $\langle X_k,\chi_{(N-\theta,\theta)}\rangle$ equals the sum of $\dm\xi$ over all $\xi\vdash k$ such that $\theta\in\Lambda^{-}(\xi)$. Put in another way,
\[
\langle X_k,\chi_{(N-\theta,\theta)}\rangle = \sum_{\xi\in\Lambda_k^+(\theta)}\dm{\xi},
\]
where recall that $\Lambda_k^+(\theta)$ is the set of all $k$-box Young diagrams that can be obtained from $\theta$ by adding at most one box per each column.
The Littlewood--Richardson rules shows that $\Spe{\theta}\otimes\Spe{(k-|\theta|)}$ induced from $\S_{|\theta|}\times\S_{k-|\theta|}$ to $\S_k$ becomes
\[
\bigoplus_{\xi\in\Lambda_k^+(\theta)}\Spe{\xi},
\]
which implies
\[
\frac{\sum_{\xi\in\Lambda_k^+(\theta)}\dm\xi}{\dm\theta\cdot1}
= \frac{|S_k|}{|\S_{|\theta|}\times\S_{k-|\theta|}|}
= \binom{k}{|\theta|},
\]
as needed.
\end{proof}

\end{document}